\documentclass[sigconf, screen]{acmart}
\AtBeginDocument{%
  }

\setcopyright{acmlicensed}
\copyrightyear{2026}
\acmYear{2026}
\acmDOI{XXXXXXX.XXXXXXX}
\acmConference[EDBT]{29th International Conference on Extending Database Technology}{March 24--27,
  2026}{Tampere, Finland}
\acmISBN{978-1-4503-XXXX-X/2018/06}

\usepackage{graphicx}
\usepackage{xmpincl}
\usepackage[table]{xcolor}
\usepackage{footnote}
\usepackage{makecell}
\usepackage{color,soul}
\usepackage{comment}
\usepackage{enumitem}
\usepackage{array}
\usepackage{listings} 
\usepackage{tikz}
\usepackage{pgfplots}
\usepackage{amsmath}
\usepackage{natbib}
\usepackage{graphicx}
\usepackage{subcaption}
\usepackage{mwe}
\usepackage[english]{babel}
\usepackage{pdfpages}
\usepackage{colortbl}
\usepackage{multirow}
\usepackage{booktabs}
\usepackage{hhline}
\pgfplotsset{compat=1.18}

\definecolor{britpurple}{RGB}{140,122,230}
\definecolor{brityellow}{RGB}{225,177,44}
\definecolor{britgreen}{RGB}{68,189,50}
\definecolor{britblue}{RGB}{25,42,86}
\definecolor{britred}{RGB}{194,54,22}
\definecolor{britred2}{RGB}{232,65,24}
\definecolor{britgreen}{RGB}{68,189,50}
\definecolor{lightgray}{gray}{0.85}
\definecolor{deepblue}{rgb}{0,0,0.5}
\definecolor{deepred}{rgb}{0.6,0,0}
\definecolor{deepgreen}{rgb}{0,0.5,0}
\definecolor{deeppurple}{rgb}{0.34,0.03,0.38}
\definecolor{lightorange}{RGB}{255,233,148}

\newcolumntype{A}{>{\raggedright\let\newline\\\arraybackslash\hspace{0pt}}m{1.9cm}}

\newcolumntype{B}{>{\raggedright\let\newline\\\arraybackslash\hspace{0pt}}m{1.3cm}}

\newcolumntype{C}{>{\columncolor{white}}c}

\begin{document}

\title{OptiMA: A Transaction-Based Framework with Throughput Optimization for Very Complex Multi-Agent Systems}

\author{Umut Çalıkyılmaz}
\affiliation{
  \institution{University of Lübeck}
  \city{Lübeck}
  \country{Germany}
}
\email{umut.calikyilmaz@.uni-luebeck.de}

\author{Nitin Nayak}
\affiliation{
  \institution{University of Lübeck}
  \city{Lübeck}
  \country{Germany}
}
\email{nitin.nayak@uni-luebeck.de}

\author{Jinghua Groppe}
\affiliation{%
  \institution{University of Lübeck}
  \city{Lübeck}
  \country{Germany}
}
\email{jinghua.groppe@uni-luebeck.de}

\author{Sven Groppe}
\affiliation{%
  \institution{University of Lübeck}
  \city{Lübeck}
  \country{Germany}
}
\email{sven.groppe@uni-luebeck.de}

\renewcommand{\shortauthors}{Çalıkyılmaz et al.}

\begin{abstract}

In recent years, the research of multi-agent systems has taken a direction to explore larger and more complex models to fulfill sophisticated tasks. We point out two possible pitfalls that might be caused by increasing complexity; susceptibilities to faults, and performance bottlenecks. To prevent the former threat, we propose a transaction-based framework to design very complex multi-agent systems (VCMAS). To address the second threat, we offer to integrate transaction scheduling into the proposed framework. We implemented both of these ideas to develop the OptiMA framework and show that it is able to facilitate the execution of VCMAS with more than a hundred agents. We also demonstrate the effect of transaction scheduling on such a system by showing improvements up to more than 16\%. Furthermore, we also performed a theoretical analysis on the transaction scheduling problem and provided practical tools that can be used for future research on it.
\end{abstract}



\received{8 October 2025}
\received[revised]{}
\received[accepted]{}

\maketitle

\section{Introduction}
\label{sec:intro}

A multi-agent system (MAS) is composed of multiple agents interacting in a common environment to achieve a goal. As a branch of distributed AI (DAI), the purpose of MAS is to find solutions to complex tasks by dividing them into simpler subproblems, each of which is assigned to a different agent \cite{durfee2001distributed}. Besides being the most natural choice for some domains, such as the coordination of a swarm of robots, MAS also offers parallelism, robustness, scalability, and the ability to use solvers with varying expertise for different parts of a problem \cite{stone2000multiagent}. Due to these advantages, MAS has been used to solve problems in various fields including robotics \cite{ota2006multi, inigo2012robotics, soriano2013multi}, cloud computing \cite{gutierrez2015agent, singh2015autonomous, nikbazm2014agent}, network security \cite{shamshirband2013appraisal, gorodetski2002multi} and routing \cite{di1998adaptive, claes2011decentralized, manvi2008multicast}.

Like many other fields, large language models (LLM) have also found their place in MAS research in recent years \cite{lu2023chameleon, liu2024dynamic, yin2023exchange}. In fact, the possibility of more intelligent LLM-based agents led academics to invest in very complex multi-agent systems (VCMAS) that are designed to accomplish sophisticated tasks, which are normally performed by a team of humans. In \cite{xu2024theagentcompany}, such a system is built to simulate a software development company using LLM-based agents as employees, and it is shown that up to 30$\%$ of jobs in such a company can be automated. In \cite{talebirad2023multi}, a general theoretical framework for the design of VCMAS is given. The proposed framework includes different agent roles, has superior-subordinate relationships between agents, allows the agents to use plugins, and allows some agents to alter the state of others on runtime. Although a framework as complex as the one proposed is definitely required to relieve human teams of such sophisticated tasks, this level of complexity is also a source of potential problems. In this paper, we address two of these potential problems, susceptibility to faults and performance bottlenecks in parallel computation.

To overcome the first problem, we propose using a transaction-based approach to design VCMAS. The idea of using the transaction concept for fault-tolerant applications outside the database context has existed for decades \cite{gray1981transaction}. In a transaction-based system, each state transition fulfills the ACID properties, preventing inconsistencies that might result from system failures, violation of constraints, or race conditions that occur due to parallel processing \cite{ramakrishnan2002database}. Transaction-based approaches for MAS have already been offered \cite{hill2006transaction, smith2003transaction}, but these approaches are not designed for a system as complex as a VCMAS. They do not have constraints to regulate the dynamic properties of such systems, and they only allow a transaction to execute actions from a single agent. Our framework integrates an improved transaction-based approach to the complex design given in \cite{talebirad2023multi}. In this way, we provide an environment to create fault-tolerant VCMAS. 

Since we applied a transaction-based approach in our framework, it is possible to treat the problem of performance bottlenecks in the context of transactions. In most modern DBMSs, multiple transactions are executed concurrently to take advantage of parallel computing. In such a process, a concurrency control (CC) method is required to satisfy the isolation property and to avoid complications that might occur when multiple transactions access the same data \cite{ramakrishnan2002database}. As we show in the following sections, our framework requires a locking-based CC, similar to a variant of 2-phase locking (2PL). Locking-based CC methods are used when a high level of isolation is required, but are known to cause performance bottlenecks due to deadlocks and idle times in threads \cite{thomasian1991performance}.

In the literature, there are many approaches to optimize the performance of DBMS with 2PL. One example is the transaction partitioning methods \cite{prasaad2020handling, zhang2018performance, sheng2019scheduling}. In these methods, conflicting transactions are grouped and processed on the same thread. Since there is no pair of transactions that are in separate groups and conflicting, execution can be done by turning CC of. Most of these methods keep a residual group to put transactions that do not fit in any group and execute residual transactions with CC. In \cite{cao2023transaction}, the authors propose scheduling the execution times of the transactions after partitioning. In this way, the number of residual transactions decreases and the total performance improves. Despite its success, this method does not consider all possible schedules and is not guaranteed to find the most efficient one. In \cite{bittner2020avoiding, groppe2021optimizing}, the problem of scheduling transactions is treated as a mathematical optimization problem, where the optimal (or near-optimal) schedule is selected from the set of all possible schedules. In such studies, the aim is not to discard CC, but to find an efficient way to execute a batch of transactions in the existence of it. In our study, we follow the example of this approach for a more analytical treatment of the problem. We examine the transaction scheduling problem (TxnSP) in depth and lay the theoretical foundations besides implementing solvers for it to be used in our framework.

The rest of the paper is formulated in the following way. Section~\ref{sec:OptiMA} introduces the proposed framework for VCMAS and explains the purposes of its design choices. In Section~\ref{sec:TxnSP} we give the formulation of TxnSP, present some theoretical findings on the problem, and introduce four optimization algorithms for it. In Section~\ref{sec:Experimental}, the results of our experiments on the complexity of TxnSP, the performance of different optimization algorithms, and the effect of transaction scheduling in OptiMA are presented. Section~\ref{sec:Conclusion} concludes the paper by summarizing our findings and suggesting some direction for future research on the subject.

\begin{table}[]
    \caption{Example operations (transactions) in a multi-agent trip planning application}
    \centering
    \scriptsize
    \begin{tabular}{B A A A }
        \toprule
         \textbf{Owner(s):} & \makecell[c]{Agent 1 \& Agent 2} & \makecell[c]{Agent 3} & \makecell[c]{Agent 4} \\
         
         \midrule
         
         \textbf{Operation:}& Finding accomodation with convenient commuting options & Buying tickets to Musée du Louvre & Buying tickets to a canal tour\\

         \midrule
         
         \multirow{4}{*}{\textbf{Instructions:}} &  - Search options & - Check calendar &  - Check calendar \\         
          & - Check commuting & - Purchase ticket & - Purchase ticket \\
          & - Log entry & - Calendar entry & - Calendar entry \\
          & & - Log entry & - Log entry \\

         \midrule

         \multirow{3}{*}{\textbf{Plugins:}} & - Search engine & - Payment & - Payment \\
         & - Log database & - Calendar & - Calendar \\
         & & - Log database & - Log database \\

         \bottomrule
    \end{tabular}
    
    \label{tab:Transactions}
\end{table}

\section{OptiMA Framework}
\label{sec:OptiMA}

In the context of multi-agent systems, our contributions can be grouped in two categories. The first is to offer a novel transaction-based framework to design VCMAS that can include up to hundreds of agents, plugins, and their relationships. The second is to provide an environment to run such systems safely and efficiently on a given number of threads. We developed the OptiMA framework that performs both these functions \footnote{\url{https://github.com/umutcalikyilmaz/OptiMA}}. In the remainder of the section, we explain these two categories of contributions, but before delving into their details, we introduce a simple example to demonstrate the purposes of our design choices.




\subsection{A Simple Example Case}
\label{ssec:Example}

Imagine a trip planning application designed with a VCMAS structure. This application utilizes five agents with distinct roles. \textit{Agent 0} is the supervisor who decides which agents to activate for a given task and assigns jobs. \textit{Agent 1} manages transportation, \textit{Agent 2} arranges accommodation, \textit{Agent 3} books museum tickets, and \textit{Agent 4} buys tour tickets. The plugins in this system are a search engine, a calendar application, a payment service, and a log database. Assume that the user asks the system to plan a trip to Paris with museum visits and a canal tour, and to make all the necessary purchases. Some possible operations for this case are given in Table~\ref{tab:Transactions}.

Each operation in Table~\ref{tab:Transactions} consists of multiple actions. If it is allowed to execute only part of an operation, we could end up with
purchased tickets that are not entered into the calendar. To avoid this, each operation should be treated as \textit{atomic}. It is also necessary to enforce certain constraints for such a system and abort any operation that does not satisfy them to ensure \textit{consistency}. For example, in the given trip application, a constraint is necessary to prevent any booking before or after the duration of the trip. Another potential source of error is the concurrent execution of transactions without proper \textit{isolation}. If the second and third transactions in Table~\ref{tab:Transactions} are executed simultaneously without any concurrency control method, it is possible that they buy tickets for overlapping times. This is due to the fact that some plugins in a multi-agent system might be \textit{non-shareable}, such as the payment and calendar plugins in the example case. Finally, transaction logging and recovery procedures are required to guarantee \textit{durability}. The execution times of complex systems might be quite long and, without necessary measures, a system crash might result in losing all progress.



\subsection{Multi-Agent System Architecture}
\label{ssec:architecture}

The framework we propose is influenced by the one given in \cite{talebirad2023multi}, but improves it by adding a few features for more flexibility and integrating the transaction concept for fault tolerance. A VCMAS designed by our framework follows a certain blueprint. This blueprint is sketched by introducing the main components and the consistency constraints of such systems in the following.

\subsubsection{Main Components}


\begin{itemize}[itemsep=1mm, leftmargin=0mm, label={}]
    \item \textbf{Plugins:} Plugins are tools offered for the use of agents. Each plugin has a single operation function that is accessible only by agents with proper authorization. Plugins are marked as either \textit{shareable} or \textit{non-shareable}. As in the example case given in Section~\ref{ssec:Example}, non-shareable plugins usually represent real-life resources and cannot be accessed by multiple agents simultaneously.
    \item \textbf{Agent Roles:} Each agent is given a role that determines its thought process and its possible actions. An agent is allowed to use operations from multiple plugins during an action. An action is also allowed to manipulate other agents and communicate with them. Having different agent roles allows the system to use agents with different expertise. In addition, system constraints, such as supervisor-subordinate relationships, depend on the role of an agent. Multiple agents can have the same role, but each agent possess its own memory which allows it to evolve separately. Since OptiMA is a dynamic system, the number of agents that have the same role is not constant. However, the initial and maximum number of agents with each role are set as parameters of a system.
    
    \item \textbf{Transactions:} A transaction in OptiMA can include actions from multiple agents. This provides the ability to create transactions with arbitrary complexity, which might even span an entire user request. The necessity of this feature is shown in the example case as the first transaction in Table~\ref{tab:Transactions} requires actions from two agents.
    A system design in our framework includes transaction templates, which are used to create transactions, similarly to an OLTP process. Having transaction templates is useful for easily determining the type of transaction and the plugins required for its execution at runtime. We note that a transaction does not have direct access to plugins, but they are used indirectly via agent actions. For each transaction template, it is also possible to define functions that are invoked in the event of a commute or abort. These functions can be used to log the transaction and define its rollback procedure. After being executed, each transaction returns a result that includes the status of the transaction (committed or aborted) and various parameters that will be used in new transactions. Each result triggers the creation of a new transaction. The specifics of this mechanism are determined during the design of transaction templates.
\end{itemize}

\subsubsection{Consistency Constraints}

These constraints draw the limits of a system by determining the permissions of each agent role and shaping its relationship with other roles. During execution, constraints are constantly checked to ensure consistency. If an action does not conform a constraint, the entire transaction is aborted. The consistency constraints in the framework are listed below.

\begin{itemize}[itemsep=1mm, leftmargin=0mm, label={}]    
    \item \textbf{Supervisor-Subordinate Relationship:} An agent role can be set as the supervisor of another role. Only an agent with supervisor role can start, stop, create, or destroy an agent of the subordinate role. The only exception is that each agent has the right to stop itself. This feature allows to design dynamic systems that would variate during model execution.
    \item \textbf{Plugin Access:} Each agent role is given the authorization to use only certain plugins.
    \item \textbf{Communication Permission:} The permission to communicate is explicitly given to certain pairs of agent roles. An agent is allowed to communicate with another agent of the same role, its supervisor, or its subordinate without the need for explicit authorization.
    \item \textbf{Halting Permission:} The ability to halt the program is granted to certain agent roles. Halting permission must be given to at least one of the roles to prevent the program from running infinitely.
\end{itemize}

\begin{figure}
    \centering
    \includegraphics[width=0.95\linewidth]{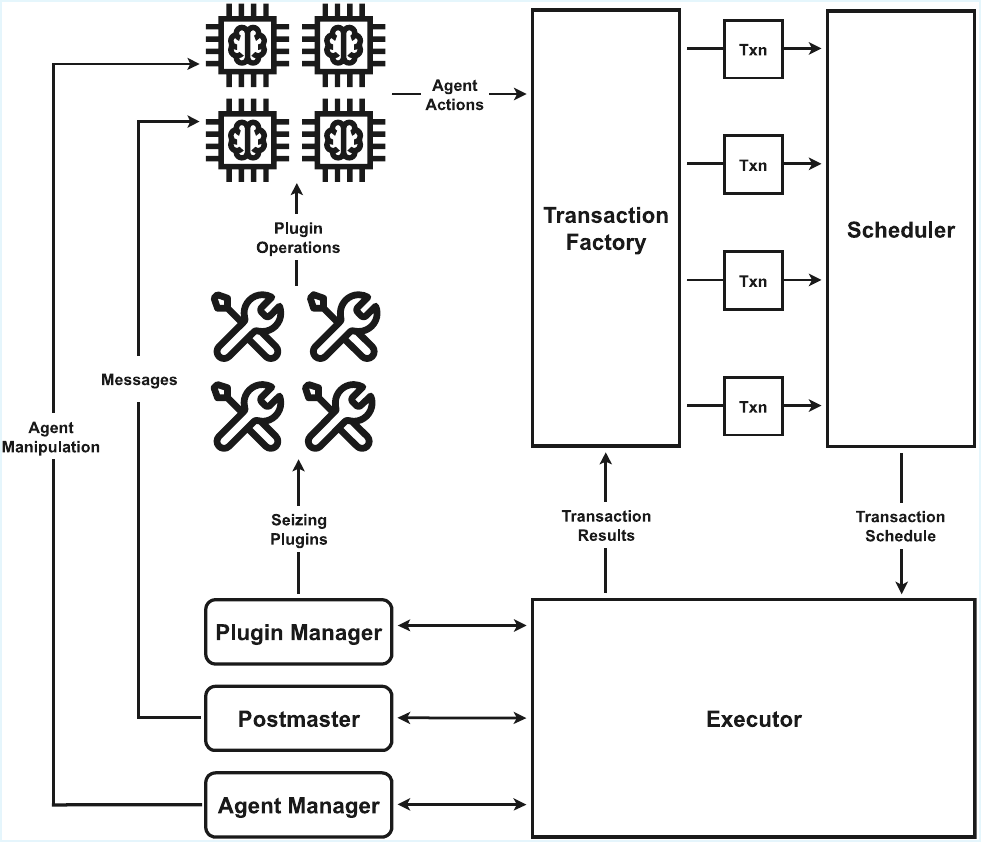}
    \caption{Structure of the Framework: The transaction factory creates transactions that includes agent actions. The scheduler batches the transactions and schedules each batch. The executor runs the scheduled transactions and sends their result back to the transaction factory, which then creates new transactions depending on these results. During the execution of a transactions, the plugin manager, agent manager and postmaster checks the consistency constraints and manipulates the system if an authorized action requests it.}
    \label{fig:OptiMA}
\end{figure}


\subsection{Model Execution}
\label{ssec:execution}

OptiMA offers centralized control over complex distributed systems to provide fault tolerance and hardware-dependent optimization. It is designed to safely execute VCMAS with hundreds of agents, even when the target hardware allows the use of a much smaller number of threads. The execution parameters and the modules used in the process are explained below.

\subsubsection{Execution Parameters}

\begin{itemize}[itemsep=1mm, leftmargin=0mm, label={}]
    \item \textbf{Optimization (boolean):} When this parameter is true, the transactions are scheduled before execution to optimize throughput.
    \item \textbf{Number of Threads ($m$):} Determines the number of parallel threads to be used for model execution.
    \item \textbf{Batch Size ($b$):} Determines the maximum number of transactions that are waiting in queue before being scheduled.
    \item \textbf{Trigger (boolean):} When this parameter is true, the transactions in queue are scheduled when one of the threads becomes idle, even if there are fewer transactions in the queue than the batch size.
\end{itemize}

\subsubsection{Execution Modules}

In Figure~\ref{fig:OptiMA}, we outline the model execution process in OptiMA. This process is explained in detail by introducing the function of each module.

\begin{itemize}[itemsep=1mm, leftmargin=0mm, label={}]    
    \item \textbf{Transaction Factory:} This module starts the execution of a model by creating the initial transactions and sending them to the scheduler. The result of each transaction is sent to the transaction factory, which is used to create a new transaction. This loop continues until the program is halted.

    
    \item \textbf{Scheduler:} If the optimization parameter is true, the first operation applied to a transaction in the scheduler is estimating its length. For this task, we implemented a simple statistical tool that keeps the average length for each transaction type. However, OptiMA also allows a system designer to insert its own estimator. After length estimation, the transaction is queued. When the number of transactions in the queue reaches $b$, or when the trigger is activated by a thread, the batch of transactions are scheduled using the optimization methods developed for TxnSP, which are explained in detail in Section~\ref{ssec:Library}. The result of a scheduling is $m$ queues of transactions, each of which will be assigned to a separate thread in the executor. If optimization is off, length estimation and scheduling steps are skipped, and the scheduler sends an arriving transaction directly to the executor to be assigned to the first idle thread.
    
    \item \textbf{Executor:} This module is responsible for the execution of transactions while ensuring isolation and consistency. When a transaction comes to the executor, it is queued in the register of one of the threads. Transactions in each register are executed according to FIFO. Before execution, a transaction is required to acquire a lock for each non-shareable plugin it uses. The details of this locking procedure are explained in Section~\ref{ssec:concurrency}. The executor constantly communicates with the plugin manager, agent manager, and postmaster to ensure consistency constraints. After execution, the result of the transaction is sent to the transaction factory.
    
    \item \textbf{Plugin Manager:} If an action wants to use a plugin, the plugin manager is notified by the executor. Then the authorization of the owner agent is checked. If authorization is verified, the plugin is seized; otherwise, the executor is told to abort the transaction.
    
    \item \textbf{Agent Manager:} When an agent action tries to start, stop, create, or destroy another agent, the executor asks the agent manager to verify permission of the owner agent. If the agent has no permission, the transaction is aborted. Otherwise, the agent manager applies the manipulation of the subject agent.
    
    \item \textbf{Postmaster:} If an agent attempts to send a message to another, the postmaster delivers it if the two agents are allowed to communicate. Otherwise, the executer is notified to abort the transaction.
\end{itemize}




\begin{figure}
    \centering
    \begin{subfigure}[b]{0.225\textwidth}
        \centering
        \begin{tikzpicture}
            \draw[->, line width=0.7] (0,0)--(0,2.2);
            \draw[->, line width=0.7] (0,0)--(3.6,0);
            
            \draw[line width=0.5] (0.1,0)--(1,1.5);
            \draw[line width=0.5] (1,1.5)--(2.3,1.5);
            \draw[line width=0.5] (2.3,1.5)--(3.3,0);

            \draw[dotted, line width=0.5] (1,1.5)--(1,0);
            \draw[dotted, line width=0.5] (2.3,1.5)--(2.3,0);

            \node at (0.6,2.33) {\footnotesize \textbf{Number}};
            \node at (0.6,2.07) {\footnotesize \textbf{of Locks}};
            
            \node at (3.6,0.27) {\footnotesize \textbf{Time}};
            
            \node at (0.55,-0.2) {\footnotesize Expanding};
            \node at (0.55,-0.46) {\footnotesize Phase};

            \node at (2.8,-0.2) {\footnotesize Shrinking};
            \node at (2.8,-0.46) {\footnotesize Phase};
        \end{tikzpicture}
        \caption{Standard 2PL}   
        \label{fig:2PL}
    \end{subfigure}
    \hfill
    \begin{subfigure}[b]{0.225\textwidth}  
        \centering 
        \begin{tikzpicture}
            \draw[->, line width=0.7] (0,0)--(0,2.2);
            \draw[->, line width=0.7] (0,0)--(3.6,0);
            
            \draw[line width=0.5] (0.1,0)--(0.1,1.5);
            \draw[line width=0.5] (0.1,1.5)--(2.2,1.5);
            \draw[line width=0.5] (2.2,1.5)--(2.2,0);

            \node at (0.6,2.33) {\footnotesize \textbf{Number}};
            \node at (0.6,2.07) {\footnotesize \textbf{of Locks}};
            
            \node at (3.6,0.27) {\footnotesize \textbf{Time}};

            \node at (2.3,-0.46) {};
        \end{tikzpicture}
        \caption{Rigorous Conservative 2PL}    
        \label{fig:RC2PL}
    \end{subfigure}
    \caption{Lock acquisition and releasing mechanisms of Two-Phase Locking and its variant} 
    \label{fig:mean and std of nets}
\end{figure}
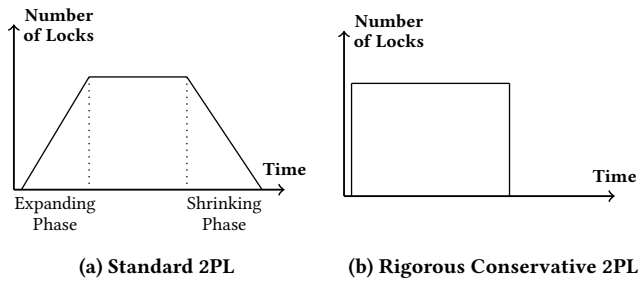

\subsection{Concurrency Control and Scheduling}
\label{ssec:concurrency}

In \cite{gray1981transaction}, actions in a transaction-based system are divided into three categories depending on their recoverability: unprotected, protected, and real. Unprotected actions are not required to be undone when a transaction is aborted, while protected actions are. Real actions are done on real objects and cannot be undone. As we show in the example in Section~\ref{ssec:Example}, it is possible that a VCMAS contains real actions. Therefore, the CC method to be used must be chosen in such a way that the number of transaction aborts is minimized. This rules out the use of optimistic concurrency control schemes, since they cause a large number of transaction aborts \cite{harder1984observations}. In OptiMA, the non-shareable plugins are treated as the critical resources of the system, and the purpose of CC is to regulate their use. Multi-version concurrency control methods do not serve this purpose, since it is not possible to keep multiple versions of most plugins \cite{korf2009multi}. We use a locking-based CC method in OptiMA, inspired by a variant of 2-phase locking (2PL) \cite{thomasian1991performance}.

In 2PL, all locks are obtained during the \textit{expanding phase}, and released during the \textit{shrinking phase} as shown in Figure~\ref{fig:2PL}. With this structure, it is possible for some actions of conflicting transactions to be processed in parallel. This is an undesirable situation for our case. In the example in Section~\ref{ssec:Example}, if the payment plugin is released before the completion of a transaction that buys museum tickets, the plugin becomes free before the time of the tickets are entered into the calendar. Then another agent might purchase tickets to another activity for an overlapping time. Thus, a CC method with a higher level of isolation is required. Rigorous conservative 2PL (RC2PL) is a variant of 2PL, in which a transaction obtains all locks before starting any action and releases them only after commit or abort (Figure~\ref{fig:RC2PL}). In this scheme, it is not allowed for any actions of conflicting transactions to coincide, so we used a similar locking procedure in our approach. There is no analogue for read or write operations in the case of plugin use. OptiMA uses a single type of lock, which is placed on a non-shareable plugin before the start of a transaction that uses it. A lock is only released after the transaction is completed and any other transaction is forbidden to use a locked non-shareable plugin. Deadlocks occur due to circular waits, which can be avoided by sequential locking \cite{gray1992transaction}. In OptiMA, each plugin is assigned a unique integer id. A transaction acquires the locks to the non-shareable plugins in the ascending order of the plugin ids to prevent deadlocks.

The type of CC used in OptiMA also determines the nature of the transaction scheduling process. The formulation of TxnSP that we use reflects the properties of the locking procedure. In the formulation of the problem, transactions are treated as non-divisible blocks and two conflicting transactions block each other completely, as we show in Section~\ref{sec:TxnSP}.

\section{The Transaction Scheduling Problem}
\label{sec:TxnSP}

\subsection{Definition of the Problem}
\label{ssec:Definition}


We define TxnSP as the problem of scheduling jobs to be processed on identical parallel machines where some pairs of jobs cannot be processed concurrently. Preemption (interruption of jobs) is not allowed and the processing time of a job can be any positive value. The objective is to find the schedule with the minimum makespan. Makespan is defined as the time required for all jobs to be completed.
Due to conflicts, a schedule can have idle times for machines.

\begin{figure}
    \centering    \includegraphics[width=0.8\linewidth]{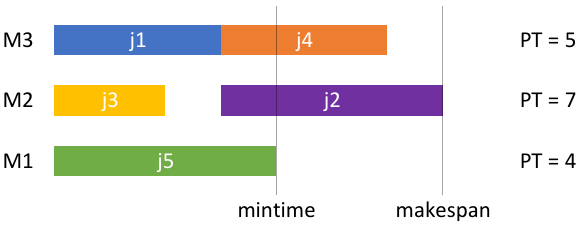}
    \caption{An example transaction schedule}
    \label{fig:schedule}
\end{figure}

In Figure~\ref{fig:schedule}, an example transaction schedule is given. In this schedule, jobs 1 and 2 are conflicting, so job 1 blocks job 2 and forces it to start later. The processing times of the machines are given on the right. The longest processing time determines the \textit{makespan}, and the shortest determines the \textit{minimum time}. The parameters used to define a TxnSP instance are given below.

\begin{itemize}[itemsep=1mm, leftmargin=0mm, label={}]
    \item $n$: Number of jobs to be scheduled
    \item $m$: Number of machines
    \item $\textbf{L}$: A $1 \times n$ matrix of job lengths
    \item $\textbf{C}$: A $n \times n$ matrix of binary values such that
    \begin{equation}
        C_{i,j} = 
        \begin{cases}
            1,& if\;i \neq j\;and\;jobs\;i\;and\;j\;conflicts \\
            0, & otherwise.
        \end{cases}
    \end{equation}
\end{itemize}

$\textbf{C}$ is a symmetric matrix with 0 values on its primary diagonal, since the conflict relation is commutative and anti-reflective.

With this definition, TxnSP can be considered as a special case of the disjunctive scheduling problem (DSP) \cite{dorndorf2000constraint}. While DSP contains both disjunctive (blocking jobs) and conjunctive (precedence relations between jobs) constraints, TxnSP is limited to include only disjunctive constraints. A problem with the same premise, namely the mutual exclusive scheduling problem (MESP), has been proposed to model the parallel solution of differential equations \cite{baker1996mutual}. In MESP it is assumed that each job has a unit length, so TxnSP is a generalization of MESP. In our formulation of TxnSP, we are inspired by \cite{bittner2020avoiding} and \cite{groppe2021optimizing}, which also assume the use of RC2PL in formulating the problem. However, those studies limit transaction lengths to integer values in problem definition and only provide a quadratic binary optimization model for the problem. In the following of this section, we continue the theoretical analysis of the TxnSP where those studies left off.

\subsection{Computational Complexity}
\label{ssec:Complexity}

We start our theoretical analysis of TxnSP by proving the complexity class of the problem.

\begin{theorem}
    TxnSP is NP-Hard.
\end{theorem}

\begin{proof}
    If an algorithm exists that solves all TxnSP instances in polynomial time, then it can also solve all MESP instances in polynomial time. It is shown in \cite{baker1996mutual} that MESP is NP-Hard. Therefore, TxnSP is also NP-Hard.
\end{proof}

Although TxnSP is NP-Hard in general, the complexity of an instance is strongly related to the nature of its conflict matrix. We consider the two extreme cases of conflict levels in the following.


\begin{definition}[Conflict Parity]
    Conflict parity is the ratio of the conflicting pair of jobs to the total number of pairs, and it is denoted as $cp$. For a problem instance with $n$ jobs, $cp$ is calculated as

    \begin{equation*}
        cp=\frac{Number\;of\;Conflicts}{C(n,2)} .
    \end{equation*}
    
    By definition, the possible values that can be assumed by a conflict parity are in the interval $[0,1]$.
\end{definition}

\begin{theorem}
    TxnSP instances with $cp=0$, are computationally intractable.
    \label{theo:cp0}
\end{theorem}

\begin{proof}
    If we limit the problem to the case of no conflicts, it trivially reduces to the multi-way partitioning problem. Since this is an NP-Hard problem, as proven in \cite{korf2009multi}, TxnSP instances with $cp=0$ are intractable.
\end{proof}

\begin{theorem}
    TxnSP instances with $cp=1$, can be solved in polynomial time.
    \label{theo:cp1}
\end{theorem}

\begin{proof}
    If all jobs conflict with each other, it is not possible to run any two jobs simultaneously. Then any schedule would be equivalent to processing all jobs one after another on a single machine. Then the makespan of the optimal schedule is calculated by summing the lengths of all jobs.
\end{proof}

The complexity of TxnSP instances varies completely between the two extremes of $cp$. This raises the question of how the problem \textit{behaves} for moderate values of $cp$. We conducted an empirical analysis of this issue by examining the structure of the solution spaces of random instances of TxnSP with various $cp$ values. The results of this analysis and their implications are presented in Section~\ref{sssec:Analysis}.

\subsection{Partial Optimality Condition}
\label{ssec:Partial}

In this section, we introduce and provide the proof of the partial optimality condition for TxnSP, which is crucial for the dynamic programming approach that we introduce in Section~\ref{ssec:Library}. We start by providing the definitions for some necessary terms.

\begin{definition}[Subschedule]
    A subschedule is a schedule that includes only a subset of the set of jobs. If a subschedule $i$ includes only the jobs in subset $S$, it is said that $i$ is formed by $S$, and the size of $i$ is $|S|$. In a subschedule each job starts at the earliest possible time.
\end{definition}

By definition, a schedule is a subschedule formed by the set of all jobs in the problem. 
Thus, all the rules about subschedules that we introduce are also valid for schedules. 


\begin{definition}[Makespan and Minimum Time]
    For subschedule $i$, makespan and minimum time are the maximum and minimum processing times of the machines, and denoted as $ms_i$ and $mt_i$, respectively.
\end{definition}


\begin{definition}[Derived Schedule]
    If the subschedule $j$ can be produced by adding $c\geq0$ number of jobs to the subschedule $i$, then $j$ is said to be \textit{derived} from $i$ and $i$ is a \textit{root} of $j$. The root of $j$ formed by subset $S$ is briefly called the $S$ root of $j$.
\end{definition}
 
By definition, in a subschedule $i$ formed by $S$, each job must start at the earliest possible time. Thus, if removing the job $\alpha \in S$ causes at least one job not to start at the earliest possible time, the result is not a subschedule and $i$ does not have a $S-\{\alpha\}$ root.

\begin{definition}[Derivation Plan]
    A derivation plan is a list of instructions that is used to derive a subschedule from another. If $j$ is derived from $i$ by adding $c$ jobs, then the related derivation plan $\mathcal{D}$ is a list of $c$ instructions. Each instruction specifies which job to be inserted into which machine.
\end{definition}

\begin{definition}[Empty Subschedule]
    The subschedule that is formed by the empty set is called the empty subschedule.
\end{definition}

\begin{definition}[Residual Subschedule]
    If the subschedule $j$ is derived from $i$ using the derivation plan $\mathcal{D}$, the residual subschedule $j-i$ is the schedule derived from the empty subschedule using $\mathcal{D}$.
\end{definition}

Each derivation plan can be represented by a residual subschedule, since a residual subschedule shows the order and place of the jobs to be inserted. We use these two terms interchangeably.

\begin{lemma}
    If the subschedule $j$ is derived from $i$, then $ms_j\leq ms_i+ms_{j-i}$.
    \label{lem:Residual}
\end{lemma}

\begin{proof}
    In the worst case scenario, the jobs in $j-i$ cannot start before all the jobs in $i$ are completed due to conflicts. In this case, $j$ will practically act as running $i$ and $j-i$, one after the other. In this case $ms_j=ms_i+ms_{j-i}$. However, equality does not always hold. In case some of the last jobs of $i$ do not conflict with the jobs in $j-i$, some of the jobs in $j$ can start earlier and it is possible that $ms_j<ms_i+ms_{j-i}$.
\end{proof}

\begin{figure}
    \centering
    \begin{subfigure}[b]{0.225\textwidth}
        \centering
        \includegraphics[scale=0.45]{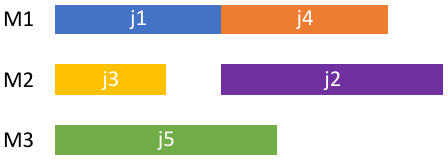}
        \caption{Base Subschedule}   
        \label{fig:lemmaa}
    \end{subfigure}
    \hfill
    \begin{subfigure}[b]{0.225\textwidth}  
        \centering
        \includegraphics[scale=0.45]{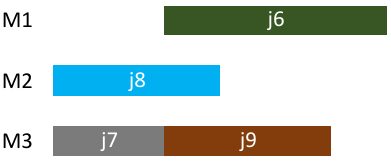}
        \caption{Residual Schedule}    
        \label{fig:lemmab}
    \end{subfigure}
    \vskip\baselineskip
    \begin{subfigure}[b]{0.45\textwidth}   
        \centering
        \includegraphics[scale=0.45]{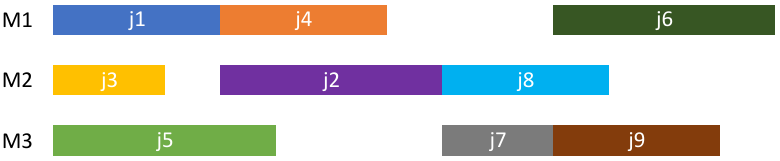}
        \caption{Worst-Case Derived Schedule}    
        \label{fig:lemmac}
    \end{subfigure}
    \vskip\baselineskip
    \begin{subfigure}[b]{0.45\textwidth}   
        \centering
        \includegraphics[scale=0.45]{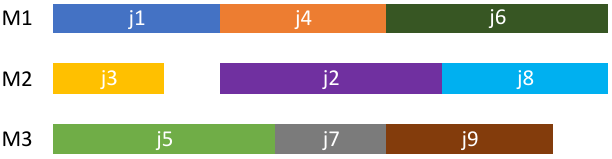}
        \caption{Best-Case Derived Schedule}   
        \label{fig:lemmad}
    \end{subfigure}
    \caption{An Example Case for Lemma~\ref{lem:Residual}}
    \label{fig:lemma}
\end{figure}

In Figure~\ref{fig:lemma}, an example case for the proof of Lemma~\ref{lem:Residual} is presented. In the worst case shown in Figure~\ref{fig:lemmac}, $j6$ and $j7$ conflict with $j2$, so the base subschedule (Figure~\ref{fig:lemmaa}) and the residual schedule (Figure~\ref{fig:lemmab}) are run one after another. In the best case shown in Figure~\ref{fig:lemmad}, each new job added to the base schedule starts immediately after the previous job in its machine.

\begin{lemma}
    If the subschedule $j$ is derived from $i$, then $ms_j\geq mt_i+ms_{j-i}$.
    \label{lem:Residual2}
\end{lemma}

\begin{proof}
    In the case where there is no conflict and the machine with the lowest processing time in $i$ is the same as the machine with the highest processing time in $j-i$, $ms_j = mt_i+ms_{j-i}$. In all other cases, $ms_j > mt_i+ms_{j-i}$
\end{proof}

\begin{definition}[Dominating Subschedule]
    \label{def:dominating}
    If the subschedules $i$ and $j$ are formed by the same subset $S$, and if $ms_i\leq mt_j$, it is said that $i$ dominates $j$. This relation is denoted as $i<j$.
\end{definition}

\begin{theorem}[Partial Optimality]
    If subschedule $i$ dominates $j$, then for each schedule $l$ formed by subset $S$ and derived from $j$, there is at least one schedule $k$ that is formed by $S$ and derived from $i$ where $ms_k\leq ms_l$.
    \label{theo:partial}
\end{theorem}

\begin{proof}
    From Lemma~\ref{lem:Residual2}, 
    the following inequality hold.

    \begin{equation}
        ms_l\geq mt_j+ms_{l-j}      
    \end{equation}

    Now assume that $k$ is derived using the same derivation plan as $l$. Then
    
    \begin{equation}
        ms_l\geq mt_j+ms_{k-i} .
        \label{eq:ineq}
    \end{equation}

    Because $i$ dominates $j$ we can write
    
    \begin{equation}
        ms_l \geq mt_j+ms_{k-i} \geq ms_i+ms_{k-i} .        
    \end{equation}

    Then, using Lemma~\ref{lem:Residual} we can show that
    \begin{equation}
        ms_l \geq ms_k .
    \end{equation}

    This shows that we can find at least one subschedule that is derived from $i$, and has a makespan smaller than or equal to any schedule derived from $j$ and concludes the proof.
\end{proof}


\subsection{Permutation Encoding}
\label{ssec:Permutation}

In \cite{ccalikyilmaz2023opportunities} it is proposed to use permutations to represent possible solutions of TxnSP. Due to its simplicity, this representation is quite useful for the implementation of various methods such as the ones we introduce in Section~\ref{ssec:Library}. In the following, we prove that permutation encoding is capable of reducing the size of a search space while also being able to represent at least one optimal solution for any problem instance.



\begin{definition}[Last Jobs]
    For subschedule $i$, the set that includes the last job processed on each machine is denoted as $LJ_i$.
\end{definition}

\begin{definition}[Starting and Completion Times]
    In a subschedule $i$, the starting and completion times of a job $\alpha$ are denoted as $st_i^{(\alpha)}$ and $ct_i^{(\alpha)}$ respectively.
\end{definition}

Since each job is processed without interruption, $ct_i^{(\alpha)}=st_i^{(\alpha)}+L_{\alpha}$ for all $\alpha \in S$ where $i$ is a subschedule formed by set $S$.

\begin{figure}
    \centering
    \includegraphics[width=0.8\linewidth]{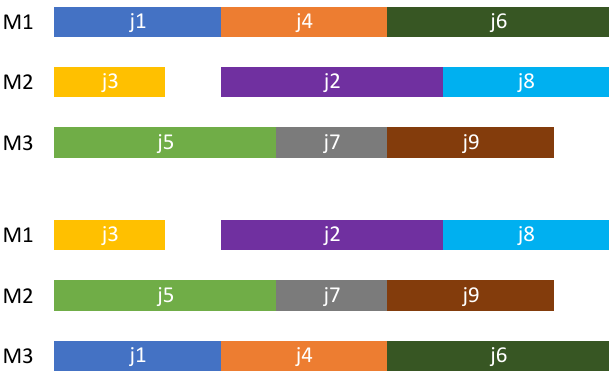}
    \caption{Two equivalent subschedules}
    \label{fig:equiv}
\end{figure}

\begin{definition}[Equivalent Schedules]
    Two subschedules $i$ and $j$ that are formed by the subset $S$ are equivalent if $LJ_i=LJ_j$ and $ct_i^{(\alpha)}=ct_j^{(\alpha)}$ for every $\alpha \in S$.
\end{definition}

Equivalency relation comes from the fact that the machines are identical in TxnSP. Thus, swapping the machines would produce the same set of processing times and the same set of last jobs, as shown in Figure~\ref{fig:equiv}. 

\begin{lemma}
    Any equivalent of an optimal schedule is also optimal.
\end{lemma}

\begin{proof}
    If a schedule $i$ is optimal, its makespan has the minimum possible value. Since swapping machines would not change the makespan, an equivalent of $i$ is also optimal.
\end{proof}

\begin{definition}[Reducible Subschedule]
    A subschedule $i$ is said to be reducible by $\alpha$ if $ms_i$ can be decreased by replacing the job $\alpha \in LJ_i$ on another machine.

\end{definition}

In Figure~\ref{fig:Reducible}, if $j6$ and $j8$ do not conflict, the subschedule is reducible, since makespan is reduced by moving $j8$ to $M2$ or $M3$.

\begin{definition}[Prime Subschedule]
    A non-reducible subschedule $i$ of size $c$ is said to be prime if and only if there exists at least one non-reducible root of $i$ of size $d$ for every $d<c$.
\end{definition}

By definition, any prime subschedule of size $c$ can be derived from the empty subschedule in $c$ steps, where in each step a new non-reducible subschedule is derived by adding one job. 

\begin{figure}
    \centering
    \includegraphics[width=0.9\linewidth]{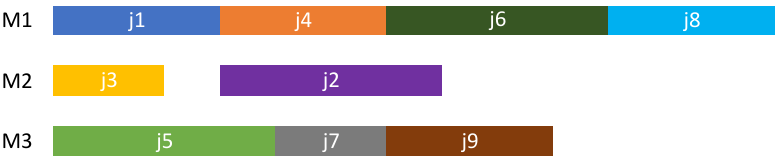}
    \caption{An Example of Reducible Subschedule}
    \label{fig:Reducible}
\end{figure}

\begin{lemma}
    Let $i$ and $j$ be subschedules that are formed by $S$. If $LJ_i=LJ_j$ and $ct_i^{(\alpha)}\leq ct_j^{(\alpha)}$ for every $\alpha\in S$, then for any subschedule $l$ derived from $j$ and formed by $S'$, there exists a subschedule $k$ that is derived from $i$ and formed by $S'$ such that $ms_k \leq ms_l$.
    \label{lem:Dominating2}
\end{lemma}

\begin{proof}
    Let us choose $k-i$ as the equivalent of $l-i$ where the machines are swapped such that the job that follows each $\alpha\in LJ_j$ is the same in both $k$ and $l$.
    In this way, it is certain that $st_k^{(\alpha)}\leq st_l^{(\alpha)}$ and, as a result, $ct_k^{(\alpha)}\leq ct_l^{(\alpha)}$ for each job $\alpha\in S'$. Since makespan is the maximum of completion times, $ms_k\leq ms_l$.
\end{proof}

\begin{lemma}
    Let $i$ be a non-reducible subschedule of size $c$ that is formed by $S$. If $i$ does not have any non-reducible root of size $c-1$, then there exists a subschedule $j$ formed by $S$ where $LJ_i=LJ_j$ and $ct_j^{(\alpha)}\leq ct_i^{(\alpha)}$ for each job $\alpha \in S$.
    \label{lem:Root}
\end{lemma}

\begin{proof}
    Let the roots of $i$ formed by $S-\{\alpha\}$ and $S-\{\beta\}$ be $j$ and $k$ respectively. If both roots exist, and if $j$ can be reduced by $\beta$, then $k$ cannot be reduced by $\alpha$. If this were not true, $i$ would be reducible in the first place. In addition, if the $S-\{\gamma\}$ root, $l$, also exists and if $k$ can be reduced by $\gamma$, then $l$ would not be reducible by neither $\alpha$ or $\beta$. Thus, in the case that all roots exist, there is at least one root that cannot be reduced by any other job in $LJ_i$.

    If no non-reducible root of size $c-1$ exists, there must be a job $\alpha \in LJ_i$ such that any root of size $c-1$ except the one formed by $S-\{\alpha\}$ can be reduced by $\alpha$, and the $S-\{\alpha\}$ root does not exist. This is only possible in the case that there is a job $\beta\in LJ_i$ such that $\alpha$ and $\beta$ conflict and $st_i^{(\alpha)}<st_i^{(\beta)}$.

    When all these conditions are satisfied, a prime subschedule $i$ can be created by the following steps.

    \begin{enumerate}
        \item Reduce the $S-\{\beta\}$ root by $\alpha$
        \item Derive a new schedule from the result of step 1 by adding $\beta$
    \end{enumerate}
    In the resulting subschedule, starting and completion times of $\alpha$ and $\beta$ are smaller than in $i$, while other jobs remain the same.
\end{proof}

\begin{figure}
    \centering
    \includegraphics[width=0.8\linewidth]{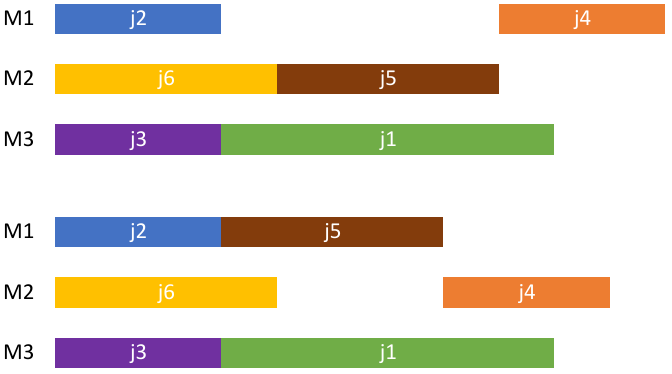}
    \caption{An example of the case explained in Lemma~\ref{lem:Root}}
    \label{fig:Prime}
\end{figure}

The subschedule shown in Figure~\ref{fig:Prime} (top) is non-reducible and does not have a non-reducible root of size $5$. If $j5$ is removed, the result is not a valid root, and removing other jobs would give reducible roots. Changing the order of $j4$ and $j5$ gives a prime subschedule (bottom) that meets the conditions in Lemma~\ref{lem:Root}.

\begin{theorem}
    Any optimal schedule is prime.
    \label{theo:Prime}
\end{theorem}

\begin{proof}
    The first part is to prove that any optimal schedule is non-reducible. If a schedule $i$ is reducible, it means that there is another schedule $j$ with the property $ms_i>ms_j$. Then $i$ is not optimal, since it does not have the smallest possible makespan.

    The second part is to prove that an optimal schedule $i$ has at least one non-reducible root of size $c$ for every $c<n$. Assume that $i$ does not have a non-reducible root of size $d$. It is proven in Lemma~\ref{lem:Root} that there exists a subschedule $l$ of size $d+1$ with property $ct_l^{(\alpha)}\leq ct_k^{(\alpha)}$ for each $d+1$ size root of $i$ and $\alpha\in LJ_k$. In that case, from Lemma~\ref{lem:Dominating2}, it is known that a schedule $j$ can be derived with the property $ms_j\leq ms_i$, and $i$ is not optimal.
\end{proof}

\begin{definition}[The Derivation Rule]
    The derivation rule is a general rule to follow when creating subschedules from permutations of jobs. This rule is stated as: \textit{Insert the next job into the machine with the smallest processing time. If more than one machine shares the smallest processing time, insert the job to the one with the smallest index (for instance, to machine 1 instead of machine 2).}
    \label{def:Derivation}
\end{definition}

The rule of selecting the machine with the lower index is just a convention. Another rule that favors the highest index for machines (or any other convention) would also be suitable, as long as one of the machines with the lowest processing time is selected.

\begin{definition}[Representable Subschedule] A subschedule that can be created from a permutation by following the derivation rule in Definition~\ref{def:Derivation} is called representable.
    
\end{definition}

\begin{lemma}
    Any representable subschedule is prime.
    \label{lem:PrimeRepresent}
\end{lemma}

\begin{proof}
    This lemma can be proven by induction.

    \begin{itemize}[leftmargin=0.4cm]
        \item Representable subschedules of size $1$ are created by adding one job to the first machine and starting it at time $0$. If the job is moved to another machine, it would still start at time $0$, so the subschedule is non-reducible and prime.
        \item Let $i$ be a prime and representable subschedule of size $c$ and $j$ be the subschedule derived by adding job $\alpha$ to $i$ following the derivation rule. In this case, $\alpha$ is already placed on the machine with the lowest processing time, so $j$ cannot be reduced by $\alpha$. Also, since $i$ is non-reducible, it cannot be  reduced by any job in $LJ_i$. Adding a new job would increase the processing time of one of the machines and does not change the reducibility of the new schedule by any of these jobs. Thus, $j$ is also prime.
    \end{itemize}
    
    \label{lem:DerivationPrime}
\end{proof}

\begin{theorem}
    The set of all representable schedules includes at least one optimal schedule.
    
    \label{theo:Derivation}
\end{theorem}

\begin{proof}
    We know that any representable subschedule is prime from Lemma~\ref{lem:PrimeRepresent}. We can use this knowledge to form a permutation from any prime subschedule or its equivalent.

    For a problem with $m$ machines, a prime subschedule of size $m$ has one job assigned to each machine. If such a subschedule $i$ of size $m$ is representable, then due to the convention of favoring machines with lower indices, it should have the property $st_i^{(\alpha)}\leq st_i^{(\beta)}$ for each pair of jobs $\alpha$ and $\beta$, where $\alpha$ is assigned to a machine with a lower index than $\beta$.

    Let $j$ be a prime subschedule with a representable root of size $m$. A permutation to represent $j$ can be found by deleting the jobs step by step, by the following rules: \textit{Delete the job that would give the non-reducible root with the highest minimum time. If there exists two such jobs, delete the one that is assigned to the machine with a higher index.} This rule is the reverse of the derivation rule and the inverse of the permutation of the deleted jobs would give the permutation representation. If there is no representable size $m$ root of $j$, an equivalent of $j$ has such a root, so it is representable. Since any optimal solution is prime, any optimal schedule or one of its equivalents is representable.

\end{proof}

\subsection{TxnSP Software Library}
\label{ssec:Library}

As a product of our theoretical work on TxnSP, we implemented a software to create, analyze, and solve TxnSP instances. Rather than integrating this software into OptiMA, we offer it as a separate library that could be used for future research on the problem \footnote{\url{https://github.com/umutcalikyilmaz/TxnSP}}. The scheduler module in OptiMA refers to this library for scheduling.

TxnSP software library has the function of creating TxnSP instances by manually entering the length and conflict matrices. It is also possible to create random problem instances by defining a distribution for the transaction lengths and entering the $cp$ value. To solve the created problem instances, four algorithms are implemented. The software package includes modules to test the accuracy and runtime of each of these algorithms and to analyze the properties of the search space for a given problem configuration. In the remainder of the section, we introduce the optimization algorithms implemented in the TxnSP software library.

\begin{figure*}[t]    \centering
    \begin{subfigure}{0.16\textwidth}
        \includegraphics[width=\textwidth]{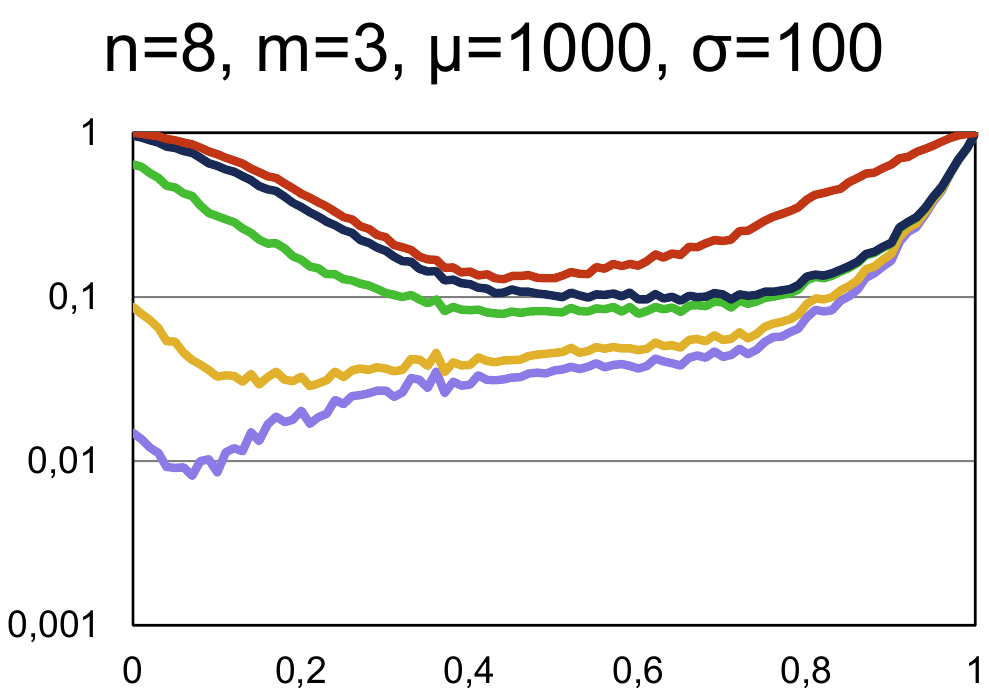}
    \end{subfigure}
    \hfill
    \begin{subfigure}{0.16\textwidth}
        \includegraphics[width=\textwidth]{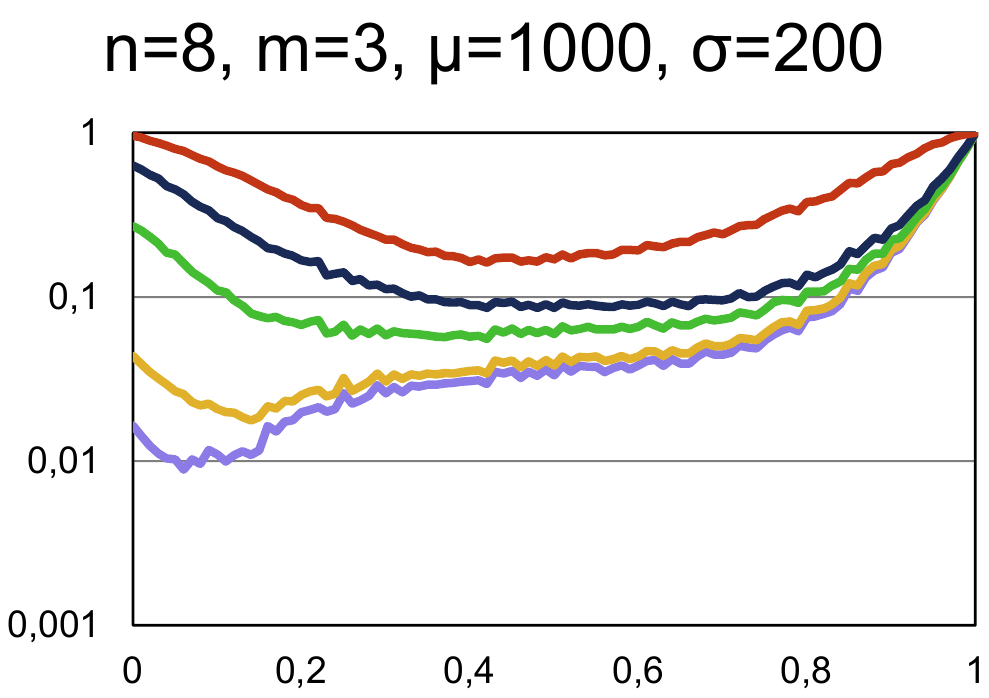}
    \end{subfigure}
    \hfill
    \begin{subfigure}{0.16\textwidth}
        \includegraphics[width=\textwidth]{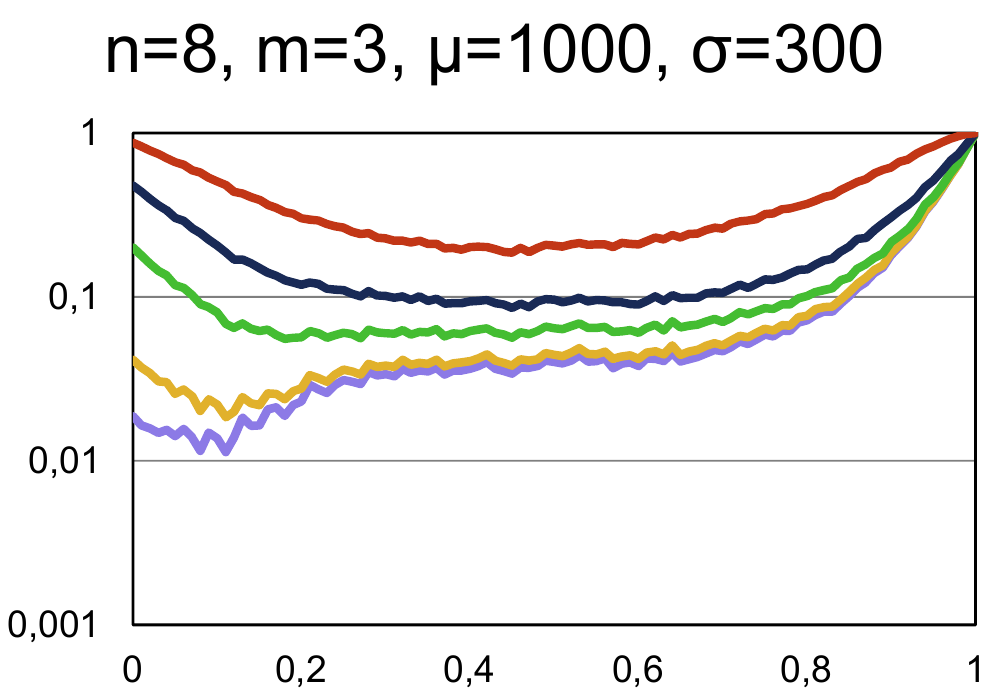}
    \end{subfigure}
    \hfill
    \begin{subfigure}{0.16\textwidth}
        \includegraphics[width=\textwidth]{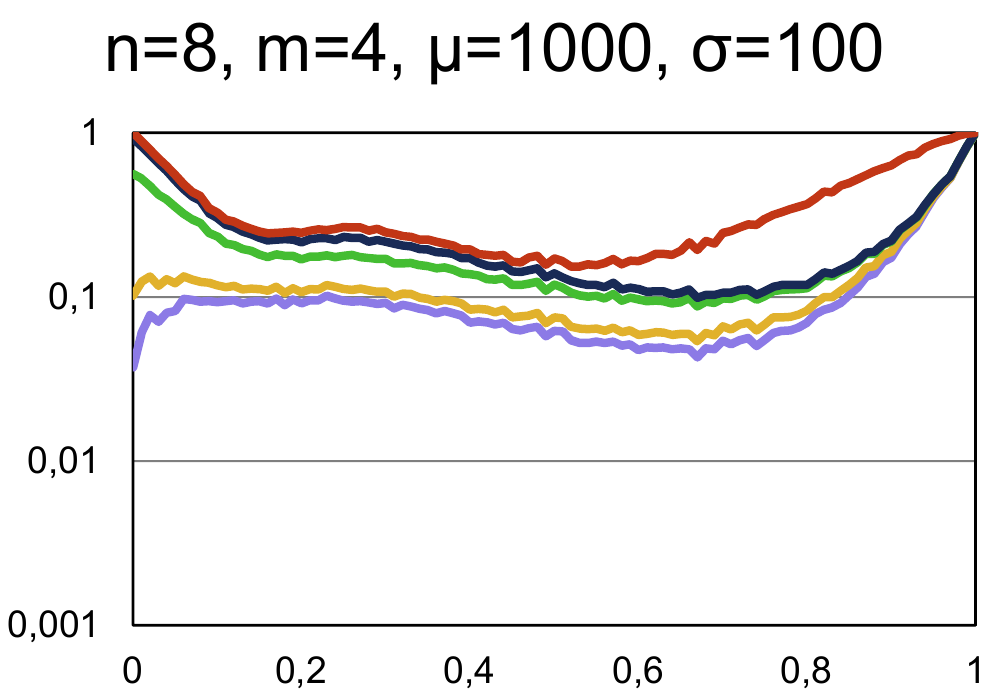}
    \end{subfigure}
    \hfill
    \begin{subfigure}{0.16\textwidth}
        \includegraphics[width=\textwidth]{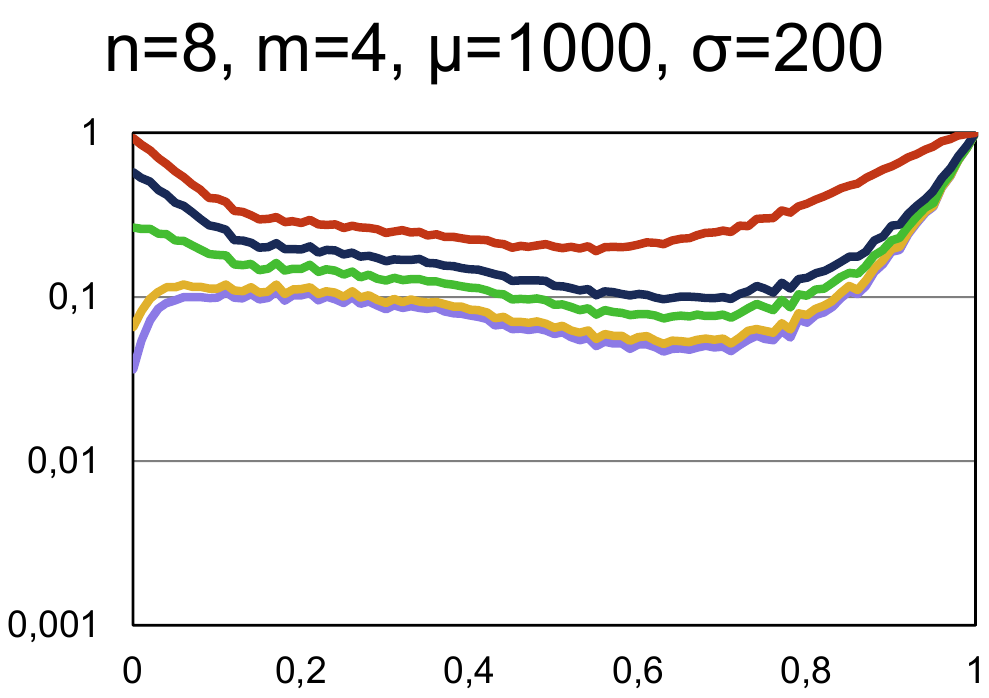}
    \end{subfigure}
    \hfill
    \begin{subfigure}{0.16\textwidth}
        \includegraphics[width=\textwidth]{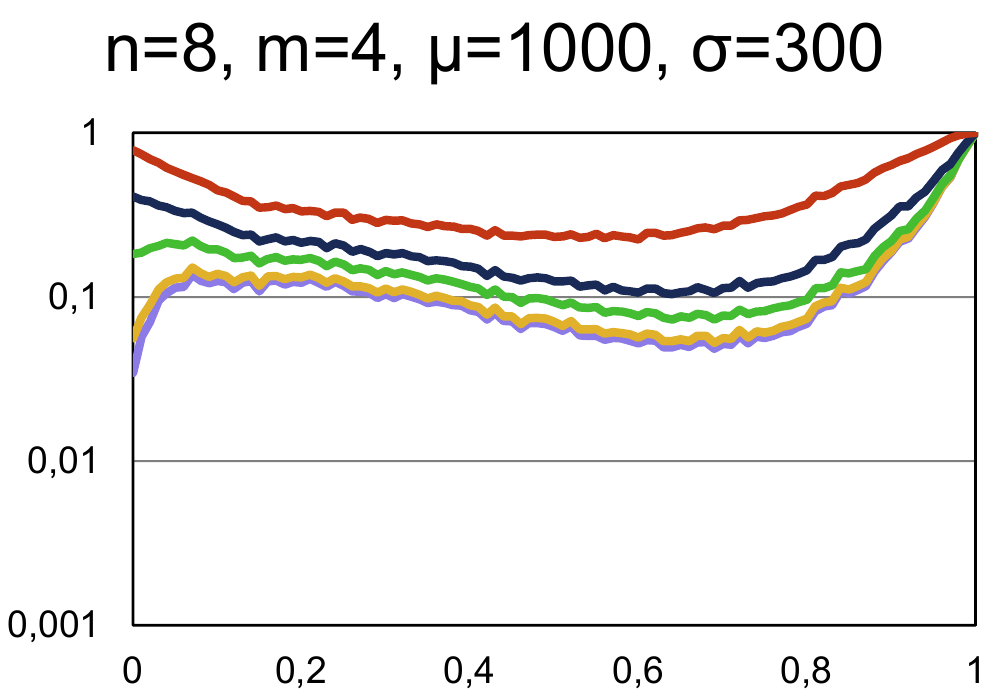}
    \end{subfigure}

    \vskip\baselineskip

    \begin{subfigure}{0.16\textwidth}
        \includegraphics[width=\textwidth]{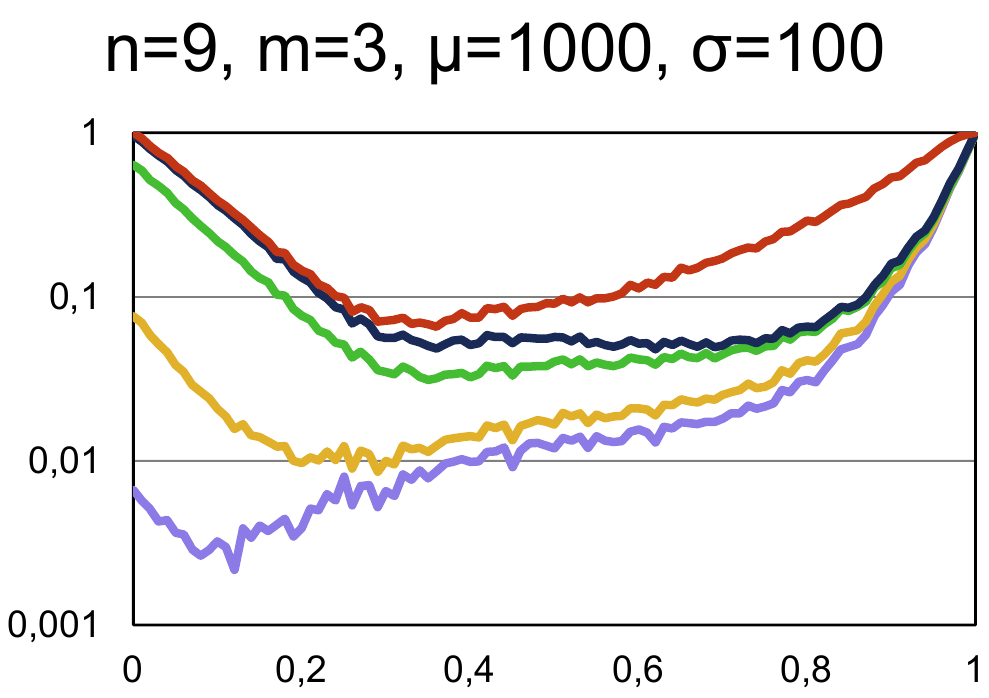}
    \end{subfigure}
    \hfill
    \begin{subfigure}{0.16\textwidth}
        \includegraphics[width=\textwidth]{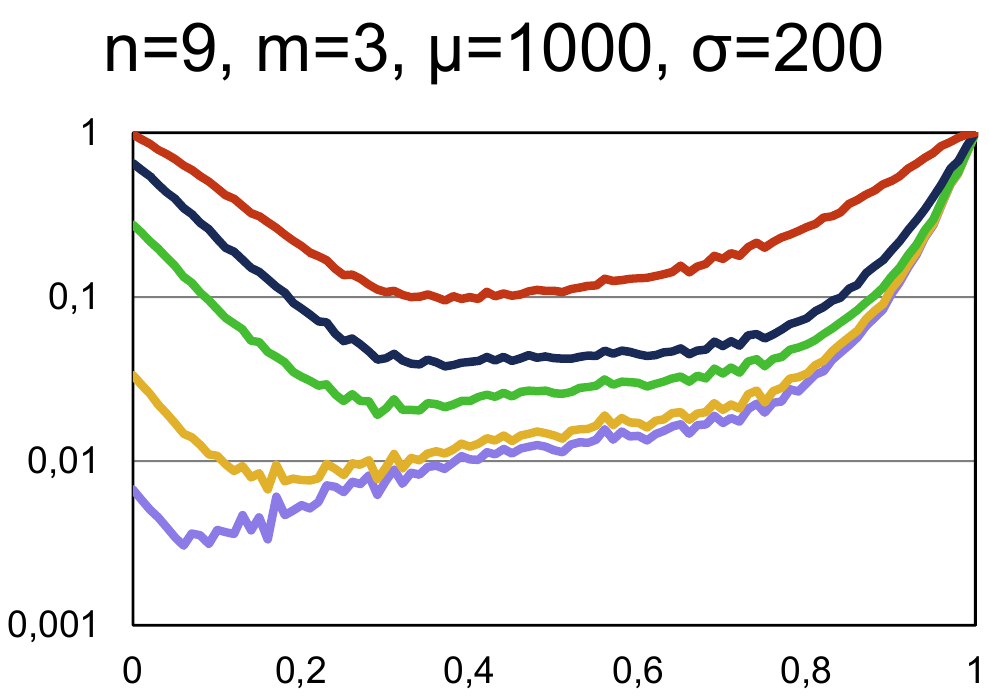}
    \end{subfigure}
    \hfill
    \begin{subfigure}{0.16\textwidth}
        \includegraphics[width=\textwidth]{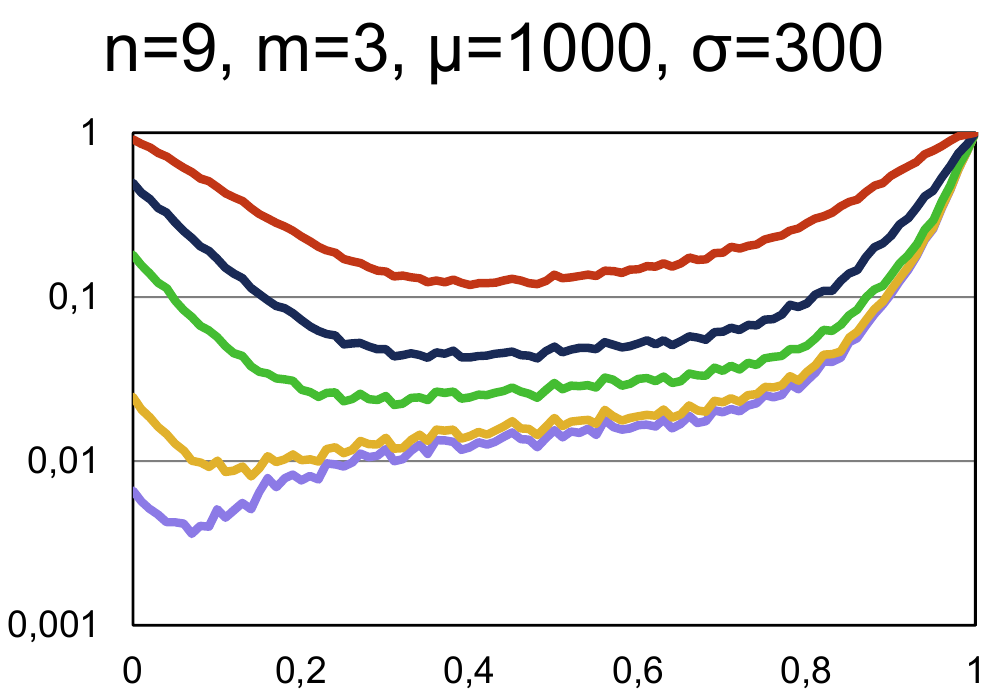}
    \end{subfigure}
    \hfill
    \begin{subfigure}{0.16\textwidth}
        \includegraphics[width=\textwidth]{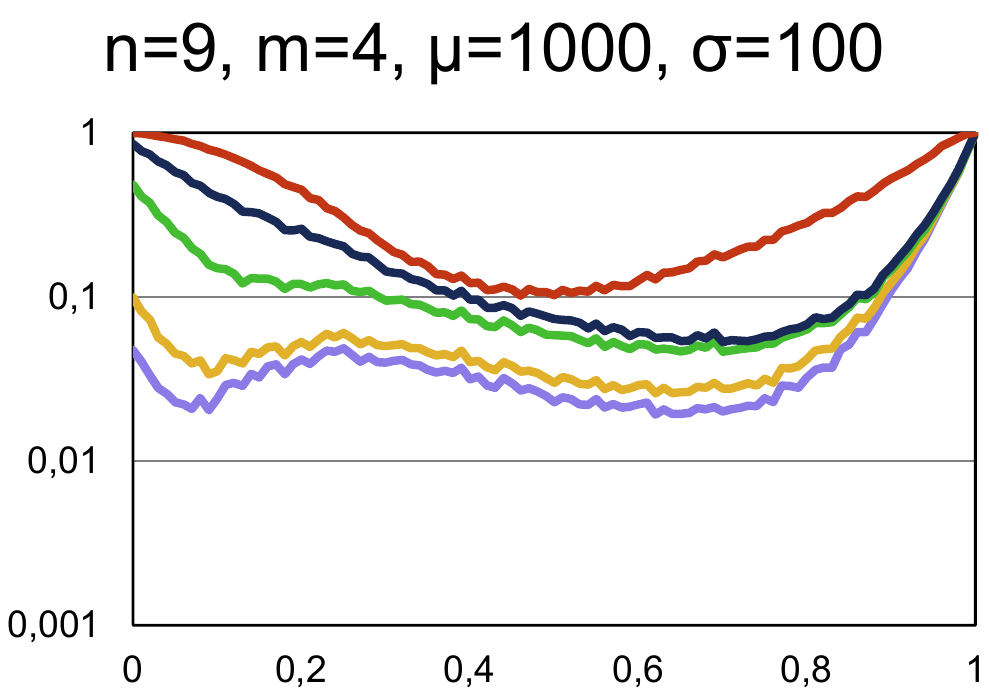}
    \end{subfigure}
    \hfill
    \begin{subfigure}{0.16\textwidth}
        \includegraphics[width=\textwidth]{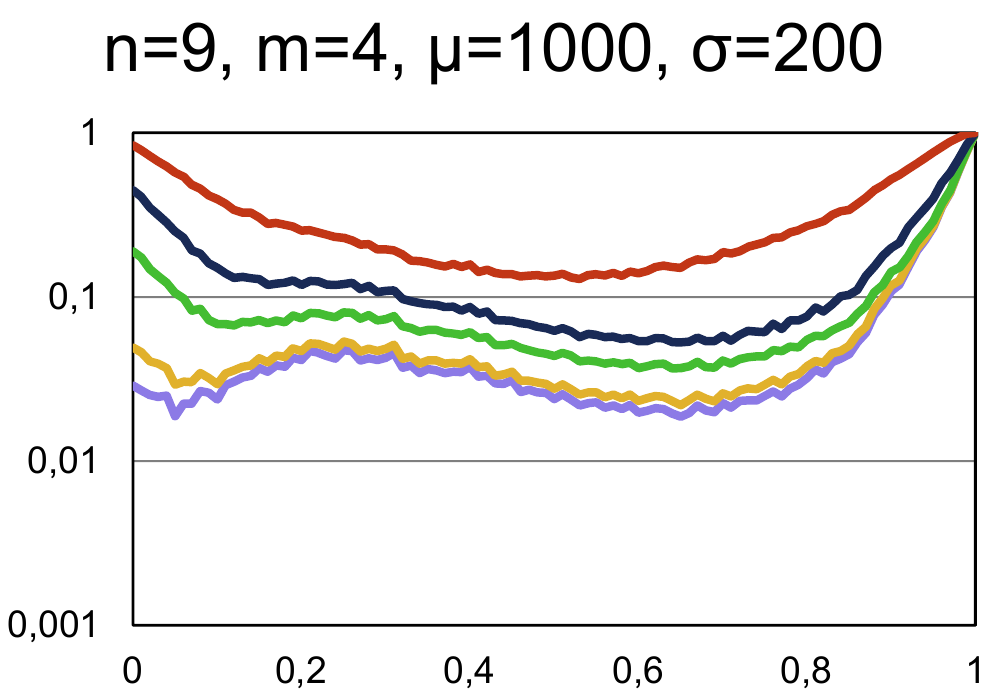}
    \end{subfigure}
    \hfill
    \begin{subfigure}{0.16\textwidth}
        \includegraphics[width=\textwidth]{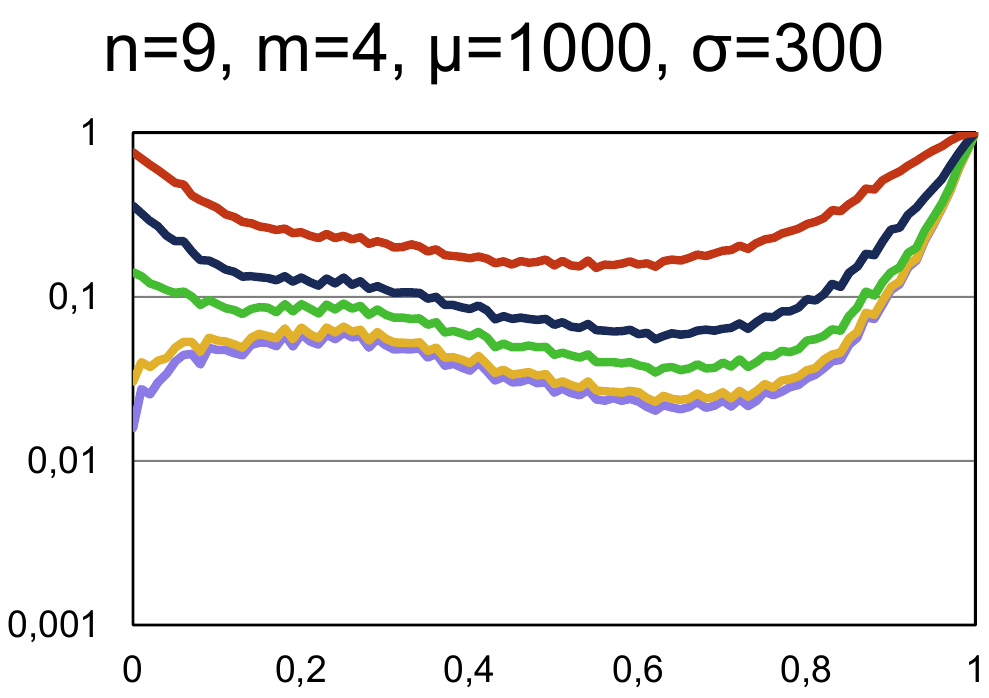}
    \end{subfigure}

    \vskip\baselineskip

    \begin{subfigure}{0.16\textwidth}
        \includegraphics[width=\textwidth]{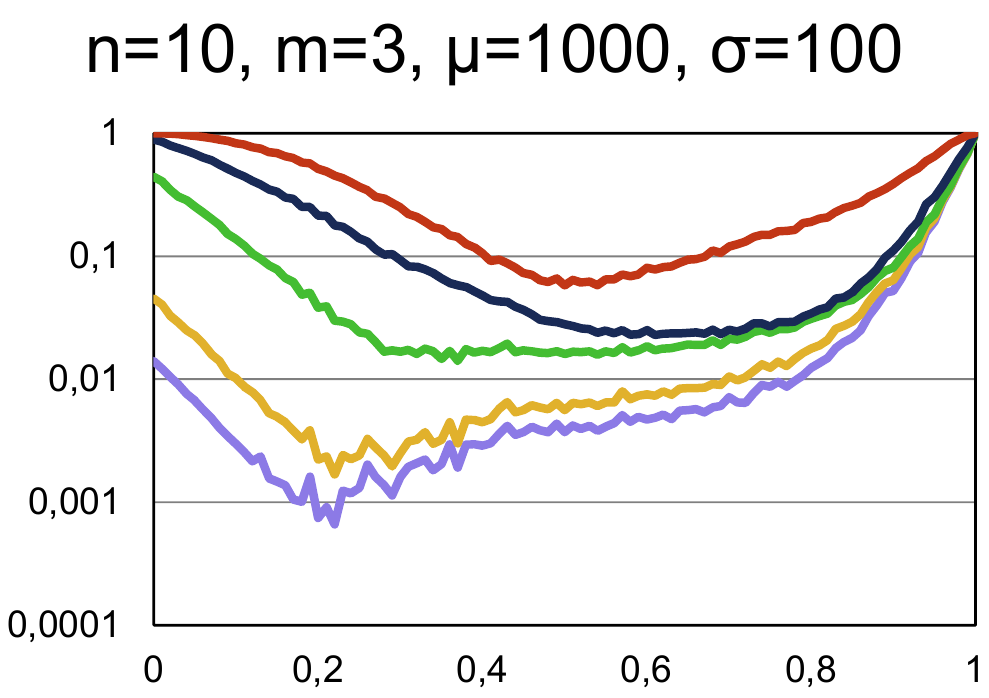}
    \end{subfigure}
    \hfill
    \begin{subfigure}{0.16\textwidth}
        \includegraphics[width=\textwidth]{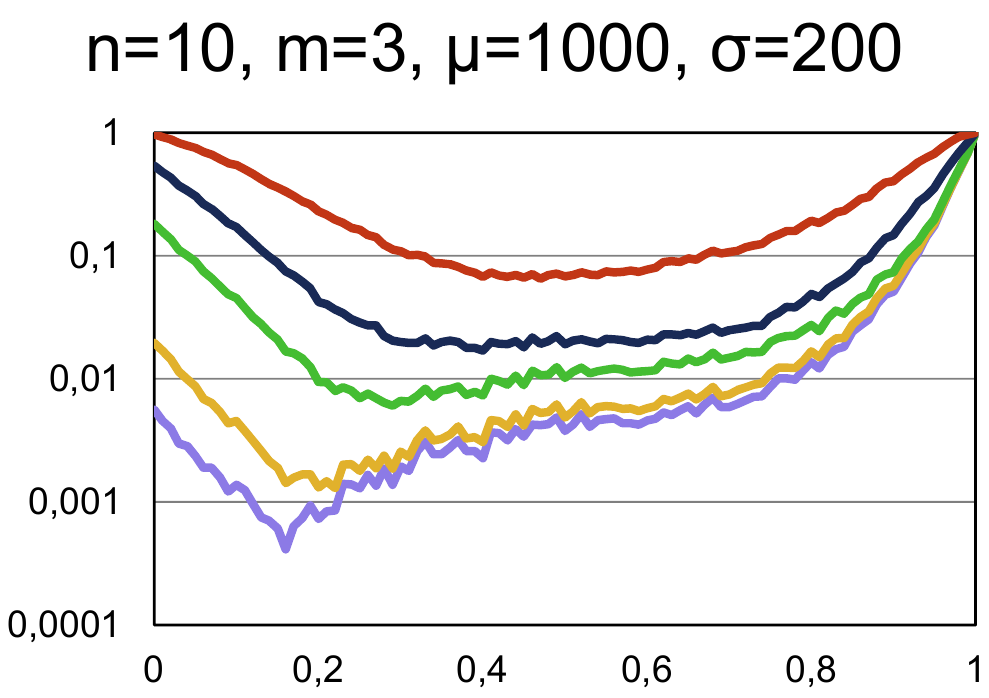}
    \end{subfigure}
    \hfill
    \begin{subfigure}{0.16\textwidth}
        \includegraphics[width=\textwidth]{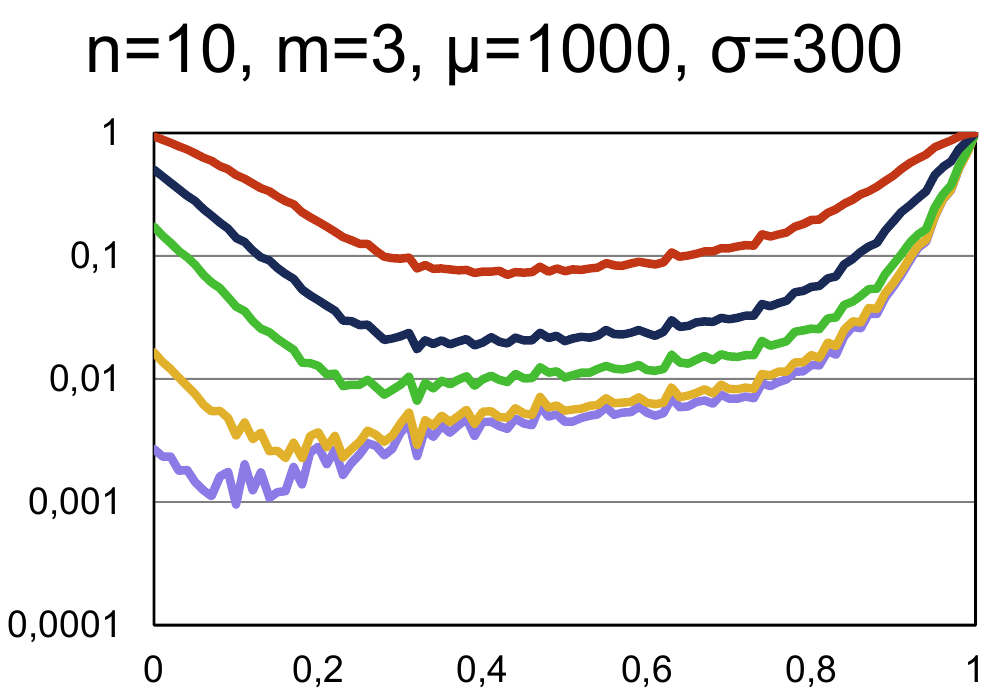}
    \end{subfigure}
    \hfill
    \begin{subfigure}{0.16\textwidth}
        \includegraphics[width=\textwidth]{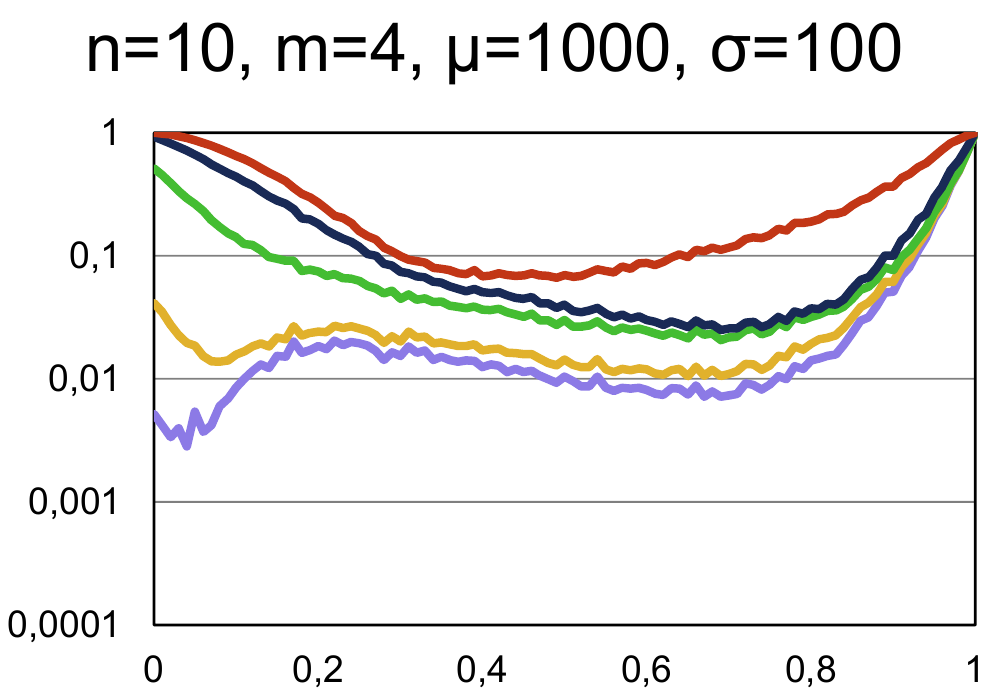}
    \end{subfigure}
    \hfill
    \begin{subfigure}{0.16\textwidth}
        \includegraphics[width=\textwidth]{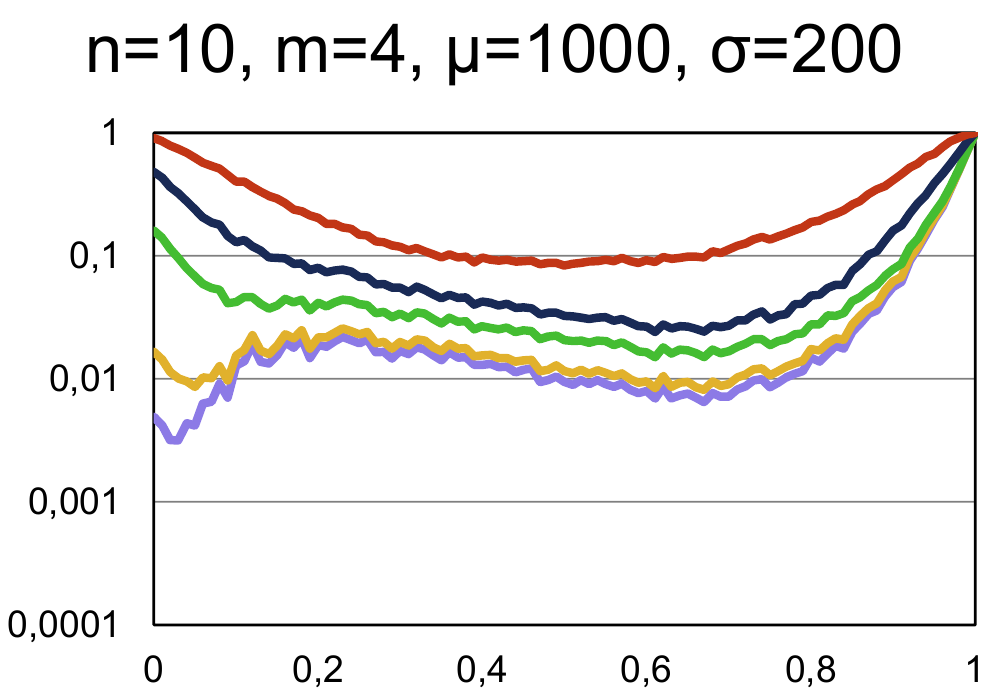}
    \end{subfigure}
    \hfill
    \begin{subfigure}{0.16\textwidth}
        \includegraphics[width=\textwidth]{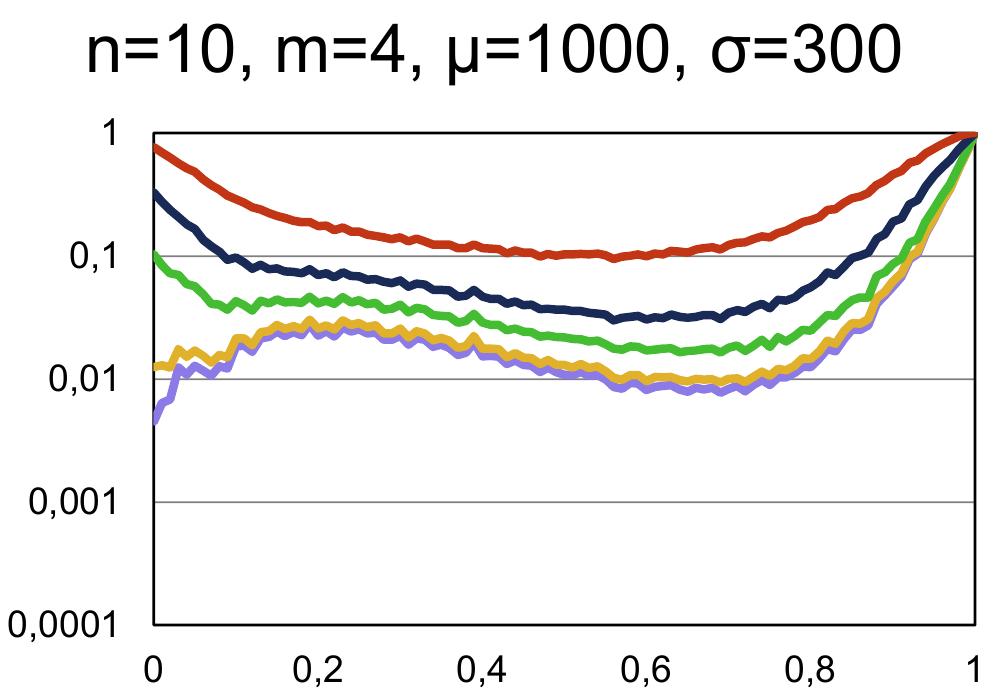}
    \end{subfigure}
    
    \caption{Analysis results of TxnSP instances: {\color{britpurple}Violet} lines shows the ratio of {\color{britpurple}optimal} schedules to the size of the solution space. {\color{brityellow}Yellow}, {\color{britgreen}Green}, {\color{britblue}Dark Blue} and {\color{britred}Red} lines show the ratio of the solutions with makespans at most {\color{brityellow}$1\%$}, {\color{britgreen}$5\%$}, {\color{britblue}$10\%$} and {\color{britred}$20\%$} higher than the optimum value, respectively. In the charts, the vertical lines show the $cp$ values, and the horizontal lines show the ratios. The chart headers shows the problem configuration, where $n$ is the number of jobs, $m$ is the number of machines, and $\mu$ and $\sigma$ are the mean and the standard deviation of the normal distribution that is used to create the length matrix.}
    \label{fig:analysis}
\end{figure*}

\subsubsection{Exhaustive Search}

The exhaustive search (ES) approach that we implemented for TxnSP evaluates all representable schedules and selects the one with the minimum makespan. The set of all representable schedules includes at least one optimal solution as shown in Theorem~\ref{theo:Derivation}, so ES is guaranteed to find an optimal solution. For an instance with $n$ jobs, each possible solution is encoded as a permutation of $n$ numbers. ES traverses $n!$ permutations, deriving a schedule from each of them using the derivation rule given in Definition~\ref{def:Derivation}, and calculates the makespan.

\subsubsection{Linear Programming}

We formulated TxnSP as a mixed-integer programming (MIP) model to allow the use of a branch-and-cut algorithm to find an exact optimal solution to the problem. In our implementation, we used the SCIP optimization suite \cite{bolusani2024scip} to solve MIP models. The detailed description of the model is given below.

\begin{itemize}[itemsep=1mm, leftmargin=0mm, label={}]
    \item \textbf{Parameters:} The problem parameters for the MIP formulation are the same as those in Section~\ref{ssec:Definition}. We introduce one more parameter, $\mathcal{M}$, which represents a very large number.
    \item \textbf{Sets:} 
    \begin{itemize}[itemsep=1mm, leftmargin=0mm, label={}]
        \item $J$: The set of all jobs, $|J|=n$
        \item $M$: The set of all machines, $|M|=m$
    \end{itemize}
    \item \textbf{Continuous Variables:}
    \begin{itemize}[itemsep=1mm,leftmargin=0mm, label={}]
        \item $ms$: The makespan of the schedule
        \item $s_i$: Starting time of the $i^{th}$ job
    \end{itemize}
    \item \textbf{Binary Variables:}
    \begin{flalign*}
        &x_{ij}=
        \begin{cases}
            1,&if\;job\;i\;is\;processed\;on\;machine\;j \\
            0,& otherwise 
        \end{cases} & \\
        \hspace{4mm}&pre_{ik}=
        \begin{cases}
            1, & if\;job\;k\;preceeds\;job\;i \\
            0, & otherwise
        \end{cases} &
    \end{flalign*}

    \item \textbf{Linear Model:}
    \begin{align}
    \label{eq:mip0} \mbox{min}~& ms & \\
    \label{eq:mip1} \mbox{s.t.  }~& ms \geq s_i + L_i  & \forall i \in \textbf{J} \\
    \label{eq:mip2} & \sum_{j = 0}^{m} s_{ij}=1 & \forall i\in \textbf{J}\\
    \label{eq:mip3}& s_i+\mathcal{M}(1-pre_{ik}) \geq s_k + L_k & \forall (i,k)\in \textbf{J}\\
    \label{eq:mip4}& pre_{ik} + pre_{ki} \geq x_{ij} + x_{kj} - 1 & \forall i\in \textbf{J},\;i>k\\
    \label{eq:mip5}& pre_{ik} + pre_{ki} \geq C_{ik} & \forall i\in \textbf{J},\;i>k\\
    \label{eq:mip6}& s_{i} \geq 0 & \forall i\in \textbf{J} \\
    \label{eq:mip7}& x_{ik}, pre_{ik} \in \{0,1\} & \forall (i,k)\in \textbf{J}
    \end{align}
    
\end{itemize}

Here, Equation~\ref{eq:mip0} shows the objective of the model, which is to minimize the makespan. By definition, the makespan is greater than or equal to the completion time of each job, and its value is imposed by the constraint in Equation~\ref{eq:mip1}. The next constraint given in Equation~\ref{eq:mip2} forces each job to be processed by exactly one machine. Constraint~\ref{eq:mip3} makes sure that job $i$ starts after job $k$ ends if $i$ is preceded by $k$. Constraints~\ref{eq:mip4} and \ref{eq:mip5} are responsible of assigning correct values to $pre_{ik}$. The first forces one of $i$ and $k$ to precede the other if they are assigned to the same machine. The latter forces one of the jobs to precede the other if they conflict. The last two constraints, \ref{eq:mip6} and \ref{eq:mip7} define the range of values that can be assigned by each variable.

\subsubsection{Dynamic Programming}

Dynamic programming (DP) can speed up the exact solution of sequencing problems by treating possible solutions as combinations instead of permutations \cite{held1962dynamic}. Since we use the permutation encoding in our implementation, TxnSP also behaves as a sequencing problem. In the first step of DP approach, subschedules of size one are created, At each step, the sizes of the subschedules are increased by adding one more job according to the derivation rule introduced in Definition~\ref{def:Derivation}. Then, partial optimality condition (Theorem~\ref{theo:partial}) is used to reduce their number. If a newly derived subschedule is dominated by another, it is eliminated. To narrow down the search space further, our approach also checks for equivalency between pairs of subschedules. In case of equivalency, one of the equivalent subschedules is eliminated. Equivalency check is only applied in the first $m$ steps.

\subsubsection{Simulated Annealing}

Due to the complexity of TxnSP, exact optimization methods require a long time even to solve moderate-size problem instances. We have implemented a simulated annealing approach that is more suitable for real-world use cases, such as optimizing multi-agent systems. This approach aims to find near-optimal solutions in a reasonable time.

SA searches through the set of neighbors of a selected point \cite{kirkpatrick1983optimization}, but the definition of a neighbor can differ between applications. Our approach uses the permutation representation of schedules, where two permutations are defined as neighbors if one can be obtained by swapping a single pair of jobs in the other. The temperature value is decreased at the end of each step and the algorithm halts when it reaches 0. We have implemented three different cooling methods, which are shown below.

\begin{align}
    &T \leftarrow T-a && [Linear\;Cooling] \label{eq:linear}\\
    &T \leftarrow rT && [Geometric\;Cooling] \label{eq:geometric}\\
    &T \leftarrow \frac{T}{1+\beta T} && [Slow\;Cooling] \label{eq:slow}
\end{align}

The variables $a$, $r$ and $\beta$ are the hyperparameters of the algorithm. Their values can be adjusted to find the best configuration for a given instance.

\section{Experimental Results}
\label{sec:Experimental}

In this section, we present the results of the experiments that we conducted using the TxnSP software library and the OptiMA framework. For these experiments, we used a system consisting of an Intel Core i7-1260P processor, 32 GB of RAM, and the Ubuntu 24.04.1 LTS operating system. We present our experiments and their results under three headings in the remainder of the section.

\subsection{Complexity Analysis}
\label{sssec:Analysis}

In Section~\ref{ssec:Complexity}, we prove that TxnSP is NP-Hard in the general case, but also show that the complexity of an instance depends on the $cp$ value. To understand the effect of $cp$ on complexity, we performed analyses on the search space of TxnSP instances.
In these analyses, we estimated the ratio of the number of optimal schedules to the size of the search space. We also estimated the ratios of schedules with makespan at most $1\%$, $5\%$, $10\%$ and $20\%$ more than the optimal value. In each analysis, these ratios are calculated for evenly spaced 101 $cp$ values in the range $[0,1]$. We use various configurations (number of jobs, number of machines, distribution of lengths) by creating 1000 random instances for each configuration. In each configuration, the lengths of the transactions are sampled from a normal distribution. We used three different standard deviation values to examine how the complexity of the search space changes depending on the shape of the distribution. The results are shown in Figure~\ref{fig:analysis}. The graphs are plotted on logarithmic scale since the ratios change in orders of magnitude.

The results show that for almost all configurations (except those with $n=8$ and $m=4$), the minimum ratio of optimal values is obtained for some point with $cp>0$. All other ratios follow this pattern, and they reach their minimum at some $cp>0$ for each configuration. Moreover, as the problem size increases (for larger $n$), the differences between the ratios at $cp=0$ and $cp>0$ also increase. This shows that the problem becomes more complex in the existence of conflicts, as the size of problem instances scales up. Another point that one might observe is that there is an inverse proportionality between the complexity and the number of machines. For higher values of $m$, all the ratios become smaller, making the search for optimal or near-optimal solutions easier. For each configuration, the ratios converge to 1 at $cp=1$. This result confirms that the problem becomes trivial when each pair of jobs conflicts, as we proposed in Theorem~\ref{theo:cp1}. Finally, we note that no significant differences are observed between the different standard deviation values for the job length distribution.


\begin{table*}[t]
    \caption{Evaluation results for random TxnSP instances with different configurations, each has length matrices created using a normal distribution with $\mu =1000$ and $\sigma=100$}
    \label{tab:Evaluation}
    \centering
    \footnotesize
    \def\arraystretch{1}
    \begin{tabular}{|c|c|c|l|c|c|c|c|c|c|c|c|c|c|c|}
         \cline{5-15}
         \multicolumn{4}{c|}{} & \multirow{4}{*}{\textbf{MIP}} & \multirow{4}{*}{\makecell{\textbf{DP}}} & \multicolumn{8}{c|}{\textbf{SA}} \\ \cline{7-14}
         
         \multicolumn{4}{c|}{} & & & \multicolumn{4}{c|}{$\mathbf{T_{max}=50}$} & \multicolumn{4}{c|}{\textbf{$\mathbf{T_{max}=100}$}} \\ \cline{7-14}
         
         \multicolumn{4}{c|}{} & & & \multicolumn{2}{c|}{\textbf{Linear}} & \multicolumn{2}{c|}{\textbf{Geometric}} & \multicolumn{2}{c|}{\textbf{Linear}} & \multicolumn{2}{c|}{\textbf{Geometric}} \\ \cline{7-14}
         
         \multicolumn{4}{c|}{}& & & $\alpha=0.1$ & $\alpha=0.01$ & $r=0.9$ & $r=0.99$ & $\alpha=0.1$ & $\alpha=0.01$ & $r=0.9$ & $r=0.99$ \\
         \hline         

         \multirow{16}{*}{$n=9$} & \multirow{8}{*}{$m=3$} & \multirow{2}{*}{$cp=0.2$} & \cellcolor{lightgray} Accuracy (\%) & \cellcolor{lightgray} 100.00 & \cellcolor{lightgray} 100.00 & \cellcolor{lightgray} 84.9 & \cellcolor{lightgray} 96.2 & \cellcolor{lightgray} 73.3 & \cellcolor{lightgray} 85.2 & \cellcolor{lightgray} 92.2 & \cellcolor{lightgray} 99.7 & \cellcolor{lightgray} 73.9 & \cellcolor{lightgray} 86.0 \\
         
         & & & Duration (ms) & 495.56 & 190.50 & 0.47 & 4.62 & 0.12 & 1.33 & 0.94 & 8.96 & 0.15 & 1.37 \\ 
         \hhline{|~|~|-|-|-|-|-|-|-|-|-|-|-|-|}
         
         & & \multirow{2}{*}{$cp=0.4$} & \cellcolor{lightgray} Accuracy (\%) & \cellcolor{lightgray} 100.00 & \cellcolor{lightgray} 100.00 & \cellcolor{lightgray} 72.1 & \cellcolor{lightgray} 79.8 & \cellcolor{lightgray} 55.4 & \cellcolor{lightgray} 71.1 & \cellcolor{lightgray} 76.1 & \cellcolor{lightgray} 84.4 & \cellcolor{lightgray} 55.9 & \cellcolor{lightgray} 70.5 \\
         
         & & & Duration (ms) & 504.26 & 158.99 & 0.48 & 4.16 & 0.19 & 1.82 & 0.94 & 9.48 & 0.20 & 1.90 \\ 
         \hhline{|~|~|-|-|-|-|-|-|-|-|-|-|-|-|}
         
         & & \multirow{2}{*}{$cp=0.6$} & \cellcolor{lightgray} Accuracy (\%) & \cellcolor{lightgray} 100.00 & \cellcolor{lightgray} 100.00 & \cellcolor{lightgray} 69.8 & \cellcolor{lightgray} 76.9 & \cellcolor{lightgray} 59.6 & \cellcolor{lightgray} 71.0 & \cellcolor{lightgray} 77.3 & \cellcolor{lightgray} 90.5 & \cellcolor{lightgray} 58.0 & \cellcolor{lightgray} 71.7 \\
         
         & & & Duration (ms) & 493.71 & 166.32 & 0.49 & 4.79 & 0.19 & 1.91 & 0.96 & 9.66 & 0.20 & 1.98\\ 
         \hhline{|~|~|-|-|-|-|-|-|-|-|-|-|-|-|}
         
         & & \multirow{2}{*}{$cp=0.8$} & \cellcolor{lightgray} Accuracy (\%) & \cellcolor{lightgray} 100.00 & \cellcolor{lightgray} 100.00 & \cellcolor{lightgray} 72.6 & \cellcolor{lightgray} 82.5 & \cellcolor{lightgray} 60.4 & \cellcolor{lightgray} 74.3 & \cellcolor{lightgray} 81.0 & \cellcolor{lightgray} 93.6 & \cellcolor{lightgray} 60.1 & \cellcolor{lightgray} 75.4 \\
         
         & & & Duration (ms) & 493.66 & 175.29 & 0.52 & 4.97 & 0.18 & 1.81 & 1.01 & 10.02 & 0.19 & 1.87 \\ 
         \hhline{|~|-|-|-|-|-|-|-|-|-|-|-|-|-|}
         
         & \multirow{8}{*}{$m=4$} & \multirow{2}{*}{$cp=0.2$} & \cellcolor{lightgray} Accuracy (\%) & \cellcolor{lightgray} 100.00 & \cellcolor{lightgray} 100.00 & \cellcolor{lightgray} 96.5 & \cellcolor{lightgray} 99.4 & \cellcolor{lightgray} 89.8 & \cellcolor{lightgray} 97.0 & \cellcolor{lightgray} 98.2 & \cellcolor{lightgray} 100.0 & \cellcolor{lightgray} 90.8 & \cellcolor{lightgray} 97.1 \\
         
         & & & Duration (ms) & 356.00 & 210.35 & 0.52 & 5.17 & 0.16 & 1.59 & 1.03 & 10.38 & 0.17 & 1.64 \\ 
         \hhline{|~|~|-|-|-|-|-|-|-|-|-|-|-|-|}
         
         & & \multirow{2}{*}{$cp=0.4$} & \cellcolor{lightgray} Accuracy (\%) & \cellcolor{lightgray} 100.00 & \cellcolor{lightgray} 100.00 & \cellcolor{lightgray} 90.2 & \cellcolor{lightgray} 94.3 & \cellcolor{lightgray} 81.0 & \cellcolor{lightgray} 92.3 & \cellcolor{lightgray} 94.8 & \cellcolor{lightgray} 96.3 & \cellcolor{lightgray} 83.3 & \cellcolor{lightgray} 91.7 \\
         
         & & & Duration (ms) & 414.65 & 157.33 & 0.51 & 5.10 & 0.13 & 1.17 & 1.01 & 10.22 & 0.13 & 1.24 \\ 
         \hhline{|~|~|-|-|-|-|-|-|-|-|-|-|-|-|}
         
         & & \multirow{2}{*}{$cp=0.6$} & \cellcolor{lightgray} Accuracy (\%) & \cellcolor{lightgray} 100.00 & \cellcolor{lightgray} 100.00 & \cellcolor{lightgray} 89.6 & \cellcolor{lightgray} 94.0 & \cellcolor{lightgray} 82.1 & \cellcolor{lightgray} 90.4 & \cellcolor{lightgray} 92.8 & \cellcolor{lightgray} 97.1 & \cellcolor{lightgray} 80.9 & \cellcolor{lightgray} 90.5 \\
         
         & & & Duration (ms) & 496.50 & 180.57 & 0.51 & 5.10 & 0.13 & 1.25 & 1.02 & 10.14 & 0.14 & 1.31 \\ 
         \hhline{|~|~|-|-|-|-|-|-|-|-|-|-|-|-|}
         
         & & \multirow{2}{*}{$cp=0.8$} & \cellcolor{lightgray} Accuracy (\%) & \cellcolor{lightgray} 100.00 & \cellcolor{lightgray} 100.00 & \cellcolor{lightgray} 91.5 & \cellcolor{lightgray} 94.9 & \cellcolor{lightgray} 83.1 & \cellcolor{lightgray} 93.8 & \cellcolor{lightgray} 93.9 & \cellcolor{lightgray} 98.0 & \cellcolor{lightgray} 85.5 & \cellcolor{lightgray} 91.4 \\
         
         & & & Duration (ms) & 439.39 & 187.15 & 0.53 & 5.07 & 0.18 & 1.82 & 1.04 & 10.27 & 0.19 & 1.89 \\ 
         \hline

         \multirow{16}{*}{$n=10$} & \multirow{8}{*}{$m=3$} & \multirow{2}{*}{$cp=0.2$} & \cellcolor{lightgray} Accuracy (\%) & \cellcolor{lightgray} 100.00 & \cellcolor{lightgray} 100.00 & \cellcolor{lightgray} 70.6 & \cellcolor{lightgray} 87.6 & \cellcolor{lightgray} 57.0 & \cellcolor{lightgray} 75.1 & \cellcolor{lightgray} 80.3 & \cellcolor{lightgray} 97.3 & \cellcolor{lightgray} 56.1 & \cellcolor{lightgray} 75.9 \\
         
         & & & Duration (ms) & 504.15 & 1518.34 & 0.40 & 3.93 & 0.17 & 1.75 & 0.79 & 7.85 & 0.18 & 1.79 \\ 
         \hhline{|~|~|-|-|-|-|-|-|-|-|-|-|-|-|}
         
         & & \multirow{2}{*}{$cp=0.4$} & \cellcolor{lightgray} Accuracy (\%) & \cellcolor{lightgray} 100.00 & \cellcolor{lightgray} 100.00 & \cellcolor{lightgray} 52.2 & \cellcolor{lightgray} 76.6 & \cellcolor{lightgray} 37.2 & \cellcolor{lightgray} 57.0 & \cellcolor{lightgray} 60.6 & \cellcolor{lightgray} 93.0 & \cellcolor{lightgray} 33.8 & \cellcolor{lightgray} 57.3 \\
         
         & & & Duration (ms) & 2215.98 & 955.39 & 0.47 & 4.61 & 0.18 & 1.79 & 0.92 & 9.34 & 0.19 & 1.86 \\ 
         \hhline{|~|~|-|-|-|-|-|-|-|-|-|-|-|-|}
         
         & & \multirow{2}{*}{$cp=0.6$} & \cellcolor{lightgray} Accuracy (\%) & \cellcolor{lightgray} 100.00 & \cellcolor{lightgray} 100.00 & \cellcolor{lightgray} 52.8 & \cellcolor{lightgray} 66.9 & \cellcolor{lightgray} 36.8 & \cellcolor{lightgray} 58.3 & \cellcolor{lightgray} 64.5 & \cellcolor{lightgray} 86.9 & \cellcolor{lightgray} 39.1 & \cellcolor{lightgray} 60.5 \\
         
         & & & Duration (ms) & 1319.94 & 1126.29 & 0.42 & 4.05 & 0.18 & 1.77 & 0.82 & 7.97 & 0.19 & 1.82\\ 
         \hhline{|~|~|-|-|-|-|-|-|-|-|-|-|-|-|}
         
         & & \multirow{2}{*}{$cp=0.8$} & \cellcolor{lightgray} Accuracy (\%) & \cellcolor{lightgray} 100.00 & \cellcolor{lightgray} 100.00 & \cellcolor{lightgray} 57.0 & \cellcolor{lightgray} 72.5 & \cellcolor{lightgray} 40.9 & \cellcolor{lightgray} 59.1 & \cellcolor{lightgray} 65.5 & \cellcolor{lightgray} 85.7 & \cellcolor{lightgray} 42.7 & \cellcolor{lightgray} 60.6 \\
         
         & & & Duration (ms) & 1007.65 & 1284.20 & 0.64 & 5.69 & 0.14 & 1.28 & 1.20 & 11.57 & 0.14 & 1.35 \\ 
         \hhline{|~|-|-|-|-|-|-|-|-|-|-|-|-|-|}
         
         & \multirow{8}{*}{$m=4$} & \multirow{2}{*}{$cp=0.2$} & \cellcolor{lightgray} Accuracy (\%) & \cellcolor{lightgray} 100.00 & \cellcolor{lightgray} 100.00 & \cellcolor{lightgray} 91.3 & \cellcolor{lightgray} 97.2 & \cellcolor{lightgray} 80.9 & \cellcolor{lightgray} 92.7 & \cellcolor{lightgray} 95.9 & \cellcolor{lightgray} 99.8 & \cellcolor{lightgray} 78.7 & \cellcolor{lightgray} 93.0 \\
         
         & & & Duration (ms) & 485.07 & 2141.06 & 0.56 & 5.30 & 0.22 & 2.16 & 1.08 & 11.18 & 0.22 & 2.22 \\ 
         \hhline{|~|~|-|-|-|-|-|-|-|-|-|-|-|-|}
         
         & & \multirow{2}{*}{$cp=0.4$} & \cellcolor{lightgray} Accuracy (\%) & \cellcolor{lightgray} 100.00 & \cellcolor{lightgray} 100.00 & \cellcolor{lightgray} 80.0 & \cellcolor{lightgray} 89.5 & \cellcolor{lightgray} 67.6 & \cellcolor{lightgray} 83.7 & \cellcolor{lightgray} 84.3 & \cellcolor{lightgray} 92.2 & \cellcolor{lightgray} 70.5 & \cellcolor{lightgray} 81.7 \\
         
         & & & Duration (ms) & 713.00 & 2083.16 & 0.61 & 6.05 & 0.19 & 2.07 & 1.22 & 11.69 & 0.20 & 2.16\\ 
         \hhline{|~|~|-|-|-|-|-|-|-|-|-|-|-|-|}
         
         & & \multirow{2}{*}{$cp=0.6$} & \cellcolor{lightgray} Accuracy (\%) & \cellcolor{lightgray} 100.00 & \cellcolor{lightgray} 100.00 & \cellcolor{lightgray} 80.6 & \cellcolor{lightgray} 86.7 & \cellcolor{lightgray} 65.9 & \cellcolor{lightgray} 81.5 & \cellcolor{lightgray} 83.6 & \cellcolor{lightgray} 93.0 & \cellcolor{lightgray} 67.8 & \cellcolor{lightgray} 83.7 \\
         
         & & & Duration (ms) & 615.49 & 2059.67 & 0.61 & 6.05 & 0.15 & 2.01 & 1.22 & 11.92 & 0.17 & 2.03 \\ 
         \hhline{|~|~|-|-|-|-|-|-|-|-|-|-|-|-|-|}
         
         & & \multirow{2}{*}{$cp=0.8$} & \cellcolor{lightgray} Accuracy (\%) & \cellcolor{lightgray} 100.00 & \cellcolor{lightgray} 100.00 & \cellcolor{lightgray} 83.7 & \cellcolor{lightgray} 89.9 & \cellcolor{lightgray} 72.0 & \cellcolor{lightgray} 86.3 & \cellcolor{lightgray} 86.5 & \cellcolor{lightgray} 96.3 & \cellcolor{lightgray} 70.8 & \cellcolor{lightgray} 88.0 \\
         
         & & & Duration (ms) & 442.37 & 2219.12 & 0.68 & 6.17 & 0.17 & 1.47 & 1.21 & 12.31 & 0.17 & 1.60 \\ 
         \hline

    \end{tabular}
    
    \label{tab:my_label}
\end{table*}

\subsection{Algorithm Evaluation}
\label{ssec:Evaluation}

As we show in Section~\ref{ssec:Library}, TxnSP software library involves four optimization algorithms and modules to evaluate their performance. We used these features to examine the average accuracy and runtime of these algorithms for various problem configurations. For each configuration, we randomly created 1000 instances. The total number of configurations is 96, including 3 different $n$ values, 2 different $m$ values, 4 different $cp$ values, and 4 different probability distributions for the job lengths, which totals 96000 problem instances. The solutions found by ES are used as a baseline to measure the accuracies of the DP, MIP, and SA approaches.

In Table~\ref{tab:Evaluation}, we present the results of the evaluation. Due to space constraints, we only include the results for instances that have length matrices created using a normal distribution with $\mu=1000$ and $\sigma=100$ and omit the results for $n=8$. The accuracy rows show the percentage of instances that a solver is able to find the same result as ES. The duration rows show the average time a method takes for a problem instance. The results show that MIP and exact DP methods are able to find the optimal solution of any instance in our sample, as expected. We also note that these two methods have been successful in finding the exact optimal solution for the remaining 78 problem configurations that we do not present here.

The exact methods take relatively long to find a solution, even for small problem instances, and their durations increase rapidly as $n$ increases. This shows the importance of approximate algorithms for TxnSP. Every set of hyperparameter we tested for SA requires a considerably shorter time than exact methods.
SA also shows consistent performance with at least 1 set of hyperparameters with more than $80\%$ precision for each problem configuration. The durations for SA are also not strongly affected by the problem size, which shows the potential for its use in real-world cases. Another notable observation is that every solver obtains higher precision values for higher $m$ values, which is consistent with the analysis results in Section~\ref{sssec:Analysis}.

\subsection{OptiMA Benchmark}
\label{ssec:Benchmark}

As our final experiment, we tested the performance of the OptiMA framework using the benchmark that we designed. The first purpose of this experiment is to observe whether OptiMA is able to successfully run a VCMAS with more than a hundred agents that are contesting for the use of non-shareable plugins on a limited hardware. The other purpose is to measure the effect of transaction scheduling on performance while running such a large model.

The benchmark module that we created for this purpose is called the Factory Floor Benchmark.
It simulates an automated factory floor that is operated by AI powered robots. The production facility consists of an assembly station and an inspection station. Each job requires a different sequence of assembly operations. Each of these operations consists of multiple actions and must be treated as atomic, thus each operation is represente by a transaction. After assembly, the job is transported to the inspection station and leaves the system after being inspected. 

The duration of each action is determined randomly from a normal distribution that represents realistic estimates. Actions are simulated by awaiting the threads. At the beginning of the benchmark, job descriptions are created by generating random sequences of operations, each of which is created as a random set of actions.
There are 9 types of actions (5 manual actions, 2 drilling actions, and 2 welding actions), and different propensities of these types can be set. It is also possible to change the simulation speed, which alters the duration of each action.

In the remainder of the section, we explain the details of the Factory Floor benchmark by introducing the plugins and agent roles, and then present the benchmark results for OptiMA.

\subsubsection{Agent Roles}

\begin{itemize}[itemsep=1mm, leftmargin=0mm, label={}]
    \item \textbf{Assembly Worker} assembles parts by performing manual work, drilling, and welding; and then places the assembled part on the conveyor belt.
    \item \textbf{Transporter} receives the assembled parts from the conveyor belt, transports them to the inspection station, and inserts them into the inspection queue.
    \item \textbf{Inspector} evaluates the quality of jobs using the QA scanner, reports the results, and sends the report to the floor manager.
    \item \textbf{Floor Manager} frequently checks the number of jobs accumulated at each station and alters the number of running agents of each role to decrease the accumulation in the queues.
\end{itemize}

\subsubsection{Plugins}

\begin{itemize}[itemsep=1mm, leftmargin=0mm, label={}]
    \item \textbf{Assembly Queue (Shareable)} contains the descriptions for the parts to be produced. 
    \item \textbf{Conveyor Belt (Shareable)} stores the assembled parts to be picked up by a transporter agent.
    \item \textbf{Inspection Queue (Shareable)} stores the parts to be inspected.
    \item \textbf{Drill Press (Non-Shareable):} is used by an assembly worker to perform a drilling action.
    \item \textbf{Welding Station (Non-Shareable)} is used by an assembly worker to perform a welding action.
    \item \textbf{QA Scanner (Non-Shareable)} is used by an inspector to perform a quality assessment of a part.
    \item \textbf{Output Bin (Shareable)} stores the completed parts after.
\end{itemize}

\subsubsection{Benchmark Results}

\begin{table}[]
    \caption{The Propensities of Assembly Actions for Different Conflict Configurations}
    \centering
    \def\arraystretch{1.3}
    \begin{tabular}{| c | c | c | c | c | c |}
        \cline{2-6}
        \multicolumn{1}{c|}{} & \multicolumn{5}{c|}{\textbf{Conflict Levels}} \\
        \cline{2-6}
        \multicolumn{1}{c|}{} & Very Low & Low & Medium & High & Very High  \\
        \hline
        Manual & 20 & 10 & 5 & 1 & 1 \\
        Drilling & 1 & 1 & 1 & 1 & 2 \\
        Welding & 1 & 1 & 1 & 1 & 2 \\
        \hline
    \end{tabular}
    
    \label{tab:propensity}
\end{table}

\begin{table*}
    \caption{Factory Floor Benchmark Results for OptiMA Framework}
    \footnotesize    
    \def\arraystrech{1}
    \begin{tabular}{|c|c|c|c|c|c|c|c|c|c|}
        \hline
        
        \multicolumn{3}{|c|}{\textbf{Model Configurations}} & \multicolumn{7}{c|}{\textbf{Execution Results}} \\
        \hline
        
        \multirow{3}{*}{\makecell{\textbf{Thread}\\\textbf{Number}}} & \multirow{3}{*}{\makecell{\textbf{Simulation}\\\textbf{Speed}}} & \multirow{3}{*}{\makecell{\textbf{Conflict}\\\textbf{Level}}} & \textbf{Non-Optimized} & \multicolumn{6}{c|}{\textbf{Optimized}} \\
        \hhline{|~|~|~|-|-|-|-|-|-|-|}
        
          & & & \multirow{2}{*}{\makecell{\textbf{Average} \\ \textbf{Throughput}}} & \multicolumn{2}{c|}{\textbf{Best Configuration}} & \multirow{2}{*}{\makecell{\textbf{Average} \\ \textbf{Throughput}}} & \multirow{2}{*}{\makecell{\textbf{Minimum}\\\textbf{Improvement}}} & \multirow{2}{*}{\makecell{\textbf{Maximum}\\\textbf{Improvement}}} & \multirow{2}{*}{\makecell{\textbf{Average}\\\textbf{Improvement}}} \\

         \hhline{|~|~|~|~|-|-|~|~|~|~|}

         &&&& \textbf{Batch Size} & \textbf{Trigger} & & & & \\
         \hline

         \multirow{15}{*}{8} & \multirow{5}{*}{2x} & \cellcolor{lightorange}Very Low & \cellcolor{lightorange}0.377 & \cellcolor{lightorange}50 & \cellcolor{lightorange}false & \cellcolor{lightorange}0.415 & \cellcolor{lightorange}\textbf{\textcolor{britgreen}{9.53\%}} & \cellcolor{lightorange}\textbf{\textcolor{britgreen}{10.70\%}} & \cellcolor{lightorange}\textbf{\textcolor{britgreen}{9.91\%}} \\

         & & Low & 0.357 & 75 & 
         false & 0.392 & \textbf{\color{britgreen}5.17\%} & \textbf{\color{britgreen}15.55\%} & \textbf{\color{britgreen}9.59\%}  \\

         & & \cellcolor{lightorange}Medium & \cellcolor{lightorange}0.307 & \cellcolor{lightorange}50 & 
         \cellcolor{lightorange}false & \cellcolor{lightorange}0.329 & \cellcolor{lightorange}\textbf{\color{britgreen}5.29\%} & \cellcolor{lightorange}\textbf{\color{britgreen}9.07\%} & \cellcolor{lightorange}\textbf{\color{britgreen}7.27\%} \\

         & & High & 0.193 & 75 & false & 0.201 & \textbf{\color{britgreen}2.48\%} & \textbf{\color{britgreen}6.29\%} & \textbf{\color{britgreen}4.05\%} \\

         & & \cellcolor{lightorange}Very High & \cellcolor{lightorange}0.171 & \cellcolor{lightorange}75 & \cellcolor{lightorange}false & \cellcolor{lightorange}0.174 & \cellcolor{lightorange}\textbf{\color{britgreen}0.90\%} & \cellcolor{lightorange}\textbf{\color{britgreen}2.80\%} & \cellcolor{lightorange}\textbf{\color{britgreen}1.85\%} \\
         \hhline{|~|-|-|-|-|-|-|-|-|-|}

         & \multirow{5}{*}{5x} & Very Low & 0.954 & 50 &  false & 1.034 & \textbf{\color{britgreen} 6.71\%} & \textbf{\color{britgreen}11.75\%} & \textbf{\color{britgreen}8.39\%} \\

         & & \cellcolor{lightorange}Low & \cellcolor{lightorange}0.896 & \cellcolor{lightorange}75 & \cellcolor{lightorange}false & \cellcolor{lightorange}0.986 & \cellcolor{lightorange}\textbf{\color{britgreen}7.18\%} & \cellcolor{lightorange}\textbf{\color{britgreen}17.72\%} & \cellcolor{lightorange}\textbf{\color{britgreen}10.01\%} \\

         & & Medium & 0.773 & 50 & false & 0.840 & \textbf{\color{britgreen}3.80\%} & \textbf{\color{britgreen}12.68\%} & \textbf{\color{britgreen}8.83\%} \\

         & & \cellcolor{lightorange}High & \cellcolor{lightorange}0.477 & \cellcolor{lightorange}75 & \cellcolor{lightorange}true & \cellcolor{lightorange}0.496 & \cellcolor{lightorange}\textbf{\color{britgreen}2.87\%} & \cellcolor{lightorange}\textbf{\color{britgreen}5.13\%} & \cellcolor{lightorange}\textbf{\color{britgreen}4.03\%} \\

         & & Very High & 0.436 & 50 & false & 0.453 & \textbf{\color{britgreen}1.69\%} & \textbf{\color{britgreen}5.83\%} & \textbf{\color{britgreen}3.89\%} \\
         \hhline{|~|-|-|-|-|-|-|-|-|-|}

         & \multirow{5}{*}{10x} & \cellcolor{lightorange}Very Low & \cellcolor{lightorange}1.850 &  \cellcolor{lightorange}75 & \cellcolor{lightorange}false & \cellcolor{lightorange}2.078 & \cellcolor{lightorange}\textbf{\color{britgreen}10.40\%} & \cellcolor{lightorange}\textbf{\color{britgreen}16.79\% }& \cellcolor{lightorange}\textbf{\color{britgreen}12.34\%} \\

         & & Low & 1.774 & 75 & false & 1.957 & \textbf{\color{britgreen}6.41\%} & \textbf{\color{britgreen}13.40\%} & \textbf{\color{britgreen}10.32\%} \\

         & & \cellcolor{lightorange}Medium & \cellcolor{lightorange}1.560 & \cellcolor{lightorange}75 & \cellcolor{lightorange}true & \cellcolor{lightorange}1.673 & \cellcolor{lightorange}\textbf{\color{britgreen}3.86\%} & \cellcolor{lightorange}\textbf{\color{britgreen}9.31\%} & \cellcolor{lightorange}\textbf{\color{britgreen}7.26\%} \\

         & & High & 0.956 & 75 & true & 1.004 & \textbf{\color{britgreen}1.36\%} & \textbf{\color{britgreen}6.96\%} & \textbf{\color{britgreen}5.09\%} \\

         & & \cellcolor{lightorange}Very High & \cellcolor{lightorange}0.866 & \cellcolor{lightorange}50 & \cellcolor{lightorange}true & \cellcolor{lightorange}0.895 & \cellcolor{lightorange}\textbf{\color{britgreen}1.76\%} & \cellcolor{lightorange}\textbf{\color{britgreen}4.96\%} & \cellcolor{lightorange}\textbf{\color{britgreen}3.33\%} \\
         \hline

         \multirow{15}{*}{12} & \multirow{5}{*}{2x} & Very Low & 0.436 & 75 & false & 0.462 & \textbf{\textcolor{britgreen}{5.47\%}} & \textbf{\textcolor{britgreen}{6.34\%}} & \textbf{\textcolor{britgreen}{5.97\%}} \\

         & & \cellcolor{lightorange}Low & \cellcolor{lightorange}0.404 & \cellcolor{lightorange}75 & \cellcolor{lightorange}true & \cellcolor{lightorange}0.423 & \cellcolor{lightorange}\textbf{\color{britgreen}0.87\%} & \cellcolor{lightorange}\textbf{\color{britgreen}7.40\%} & \cellcolor{lightorange}\textbf{\color{britgreen}4.79\%} \\

         & & Medium & 0.335 & 75 & false & 0.346 & \textbf{\color{britgreen}0.70\%} & \textbf{\color{britgreen}6.60\%} & \textbf{\color{britgreen}3.22\%} \\

        & & \cellcolor{lightorange}High & \cellcolor{lightorange}0.199 & \cellcolor{lightorange}50 & \cellcolor{lightorange}false & \cellcolor{lightorange}0.205 & \cellcolor{lightorange}\textbf{\color{britgreen}0.67\%} & \cellcolor{lightorange}\textbf{\color{britgreen}5.44\%} & \cellcolor{lightorange}\textbf{\color{britgreen}2.73\%} \\

        & & Very High & 0.177 & 50 & false & 0.179 & \textbf{\color{britred}-0.13\%} & \textbf{\color{britgreen}1.74\%} & \textbf{\color{britgreen}0.99\%} \\
        \hhline{|~|-|-|-|-|-|-|-|-|-|}

        & \multirow{5}{*}{5x} & \cellcolor{lightorange}Very Low & \cellcolor{lightorange}1.080 & \cellcolor{lightorange}75 & \cellcolor{lightorange}false & \cellcolor{lightorange}1.147 & \cellcolor{lightorange}\textbf{\color{britgreen}5.91\% }& \cellcolor{lightorange}\textbf{\color{britgreen}7.56\%} & \cellcolor{lightorange}\textbf{\color{britgreen}6.28\%} \\

        & & Low & .000 & 75 & true & 1.032 & \textbf{\color{britred}-0.78\%} & \textbf{\color{britgreen}5.23\%} & \textbf{\color{britgreen}3.31\%} \\

        & & \cellcolor{lightorange}Medium & \cellcolor{lightorange}0.836 & \cellcolor{lightorange}75 & \cellcolor{lightorange}false & \cellcolor{lightorange}0.854 & \cellcolor{lightorange}\textbf{\color{britred}-0.80\%} & \cellcolor{lightorange}\textbf{\color{britgreen}4.23\%} & \cellcolor{lightorange}\textbf{\color{britgreen}2.22\%} \\

        & & High & 0.494 & 50 & true & 0.502 & \textbf{\color{britgreen}0.02\%} & \textbf{\color{britgreen}3.20\%} & \textbf{\color{britgreen}1.57\%} \\

        & & \cellcolor{lightorange}Very High & \cellcolor{lightorange}0.444 & \cellcolor{lightorange}50 & \cellcolor{lightorange}false & \cellcolor{lightorange}0.450 & \cellcolor{lightorange}\textbf{\color{britred}-2.00\%} & \cellcolor{lightorange}\textbf{\color{britgreen}4.27\%} & \cellcolor{lightorange}\textbf{\color{britgreen}1.15\%} \\
        \hhline{|~|-|-|-|-|-|-|-|-|-|}

        & \multirow{5}{*}{10x} & Very Low & 2.180 & 75 & false & 2.289 & \textbf{\color{britgreen}2.45\%} & \textbf{\color{britgreen}8.23\%} & \textbf{\color{britgreen}5.05\%} \\

        & & \cellcolor{lightorange}Low & \cellcolor{lightorange}2.011 & \cellcolor{lightorange}75 & \cellcolor{lightorange}false & \cellcolor{lightorange}2.117 & \cellcolor{lightorange}\textbf{\color{britgreen}1.42\%} & \cellcolor{lightorange}\textbf{\color{britgreen}9.43\%} & \cellcolor{lightorange}\textbf{\color{britgreen}5.29\%} \\

        & & Medium & 1.678 & 50 & true & 1.722 & \textbf{\color{britgreen}0.31\%} & \textbf{\color{britgreen}4.53\%} & \textbf{\color{britgreen}2.64\%} \\

        & & \cellcolor{lightorange}High & \cellcolor{lightorange}1.006 & \cellcolor{lightorange}50 & \cellcolor{lightorange}false & \cellcolor{lightorange}1.019 & \cellcolor{lightorange}\textbf{\color{britred}-0.38\%} & \cellcolor{lightorange}\textbf{\color{britgreen}1.91\%} & \cellcolor{lightorange}\textbf{\color{britgreen}1.28\%} \\

        & & Very High & 0.892 & 50 & true & 0.906 & \textbf{\color{britgreen}0.77\%} & \textbf{\color{britgreen}1.95\%} & \textbf{\color{britgreen}1.62\%} \\
        \hline
    \end{tabular}
    \label{tab:Benchmark}
\end{table*}

In our benchmark, we set the initial numbers for assembly worker, transporter, and inspector agents as 100, 35 and 35, respectively. Their maximum numbers are set to 120, 50 and 50. We ran the benchmark for different simulation speeds (x2, x5 and x10) and different thread numbers (8 and 12). We also tried different propensities for the operations to examine the effect of the level of conflict on performance, as shown in Table~\ref{tab:propensity}. Since drilling and welding operations require the use of non-shareable plugins, increasing their propensity against manual operations increases the number of transaction pairs that block each other.

For each model configuration, we created five sets each having 500 jobs and ran the framework without optimization for each set. We also ran the simulation using optimization with different batch sizes (50, 75) and trigger settings (true, false) and determined the best setting considering the average throughput. In this context, throughput is defined by the number of completed jobs in the simulation per second.
For optimization, we used SA with linear temperature decrease, setting the maximum temperature to 100 and the decrement parameter to 0.05. We recorded the amount of improvement obtained by optimization for each set of jobs. In Table~\ref{tab:Benchmark}, we show the benchmark results. The table contains the results of the best optimization configuration and shows the minimum, maximum, and average improvements recorded for five sets of jobs.

Although the benchmark is uncomplicated with its rule-based agents, the configurations are chosen specifically to cause a high level of contention with a large number of agents and a much smaller number of threads. This setting is useful to test the ability of OptiMA to provide a fault-tolerant environment for large VCMAS. For all configurations and all sets of jobs, OptiMA was able to run the system without any inconsistencies. In addition, no deadlocks are encountered during any of these runs.

The average improvement provided by scheduling for each configuration is positive, and even the minimum improvements are positive for almost all cases.
This result shows the effectiveness of TxnSP based optimization for VCMAS. It is apparent from the results that for a smaller number of threads, scheduling is more effective. This is an expected outcome, as efficient assignment becomes more important when computing resources are more scarce. The amount of improvement is also correlated to the level of conflict in the system. For lower levels of conflict, it is possible to further parallelize the system, so there is more room for improvement. We did not find a pattern showing that the amount of improvement is directly related to the simulation speed. This shows that TxnSP-based optimization has the potential to be equally effective for different transaction lengths. Finally, we note that there is no dominant optimization configuration that gives the best results for any model. It is good practice to try different settings and choose the best one for a given model.

\section{Conclusion and Future Research Directions}
\label{sec:Conclusion}

Multi-agent systems have the potential to undertake many sophisticated tasks that used to be thought as impossible to automate. In this paper, we provide a fault-tolerant framework for these very complicated systems using the transaction concept and an engine to facilitate their execution in the presence of a large number of agents and resource contention. We have put these ideas into practice by developing the OptiMA framework, and empirically shown that it is capable of successfully running a benchmark case simulating a large and complex multi-agent model. 
In addition, we show the effectiveness of transaction scheduling for such models by recording significant amounts of improvement in throughput for the benchmark case. During our study to reach these results, we have also made a very necessary theoretical groundwork for 
the transaction scheduling problem, and developed a software library for its study. We are confident that these findings will prove useful in future research in multi-agent systems and possibly in other fields that can benefit from transaction scheduling.

Using the proposed framework to develop applications that would solve intricate tasks with practical use is beyond the scope of this paper. However, we consider such endeavors to be a natural subject of future research. With its proven success, OptiMA has the potential to be a very useful tool for such studies. We think that further research on transaction-based multi-agent systems will lead to the development of other frameworks inspired by OptiMA. In addition to this, our study on the transaction scheduling problem provides a basis for research on improving transaction handling in database management systems and any other systems that utilize the transaction concept.

\begin{acks}
    This work is funded by the German Federal Ministry of Research, Technology and Space within the funding program quantum technologies --- from basic research to market --- contract number 13N16090.
\end{acks}

\section*{Artifacts}

We list the two software artifacts that support our work in this paper by providing environment for evaluating, testing, and benchmarking. These artifacts are listed below.

\begin{enumerate}
    \item \textbf{TxnSP Software Library:}
    A C++ library developed to create, analyze, and solve TxnSP instances. TxnSP Software Library is publicly available on \href{https://github.com/umutcalikyilmaz/TxnSP}{GitHub}.
    \item \textbf{OptiMA Framework:} A framework developed to design fault-tolerant very complex multi-agent systems and run them with throughput optimization. OptiMA Framework is publicly available on \href{https://github.com/umutcalikyilmaz/TxnSP}{GitHub}.
\end{enumerate}

Both artifacts are developed to work in a Linux environment. The setup instructions and example scripts explaining how to use the artifacts are provided in the README files.

\balance
\bibliographystyle{ACM-Reference-Format}
\bibliography{references}

\end{document}